\documentclass[journal]{IEEEtran}
\ifCLASSINFOpdf
\else
\fi

\usepackage[numbers,sort&compress]{natbib}
\usepackage{amsfonts}
\usepackage{amsthm}
\usepackage{amsmath}
\usepackage{amssymb}
\usepackage{subeqnarray}
\usepackage{cases}
\usepackage[ruled,linesnumbered]{algorithm2e}
\usepackage{algorithmic}
\usepackage{bm}
\usepackage{flushend,cuted}
\usepackage{multicol}
\usepackage[numbers,sort&compress]{natbib}
\usepackage{extarrows}
\usepackage{color}
\usepackage{array}

\usepackage{makecell,diagbox}
\usepackage{multirow}
\usepackage{balance}
\usepackage{graphicx}
\usepackage{subfigure}
\usepackage{mathrsfs}
\usepackage{epstopdf}

\newtheorem{theorem}{Theorem}
\newtheorem{remark}{Remark}
\newtheorem{lemma}{Lemma}
\newcommand{\PreserveBackslash}[1]{\let\temp=\\#1\let\\=\temp}
\newcolumntype{C}[1]{>{\PreserveBackslash\centering}p{#1}}
\newcolumntype{R}[1]{>{\PreserveBackslash\raggedleft}p{#1}}
\newcolumntype{L}[1]{>{\PreserveBackslash\raggedright}p{#1}}

\begin{document}
%
\title{Joint Rate and Resource Allocation in Hybrid Digital-Analog Transmission over Fading Channels}

\author{\IEEEauthorblockN{Xiaoda Jiang, Hancheng Lu~\IEEEmembership{Member,~IEEE}}
\thanks{Copyright (c) 2015 IEEE. Personal use of this material is permitted. However, permission to use this material for any other purposes must be obtained from the IEEE by sending a request to pubs-permissions@ieee.org.}
\thanks{This work was supported in part by the National Science Foundation of China (No.91538203, 61390513, 61771445) and the Fundamental Research Funds for the Central Universities. The review of this paper was coordinated by D. Marabissi. (\em{Corresponding author: Hancheng Lu}.)}
\thanks{X. Jiang and H. Lu are with the Department of Electronic Engineering and Information Science, University of Science and Technology of China, Hefei, Anhui, China, 230027 (email: jxd95123@mail.ustc.edu.cn, hclu@ustc.edu.cn).}
}


\maketitle

\begin{abstract}
In hybrid digital-analog (HDA) systems, resource allocation has been utilized to achieve desired distortion performance. However, existing studies on this issue assume error-free digital transmission, which is not valid for fading channels. With time-varying channel fading, the exact channel state information is not available at the transmitter. Thus, random outage and resulting digital distortion cannot be ignored. Moreover, rate allocation should be considered in resource allocation, since it not only determines the amount of information for digital transmission and that for analog transmission, but also affects the outage probability. Based on above observations, in this paper, we attempt to perform joint rate and resource allocation strategies to optimize system distortion in HDA systems over fading channels. Consider a bandwidth expansion scenario where a memoryless Gaussian source is transmitted in an HDA system with the entropy-constrained scalar quantizer (ECSQ). Firstly, we formulate the joint allocation problem as an expected system distortion minimization problem where both analog and digital distortion are considered. Then, in the limit of low outage probability, we decompose the problem into two coupled sub-problems based on the block coordinate descent method, and propose an iterative gradient algorithm to approach the optimal solution. Furthermore, we extend our work to the multivariate Gaussian source scenario where a two-stage fast algorithm integrating rounding and greedy strategies is proposed to optimize the joint rate and resource allocation problem. Finally, simulation results demonstrate that the proposed algorithms can achieve up to 2.3dB gains in terms of signal-to-distortion ratio over existing schemes under the single Gaussian source scenario, and up to 3.5dB gains under the multivariate Gaussian source scenario.
\end{abstract}


\IEEEpeerreviewmaketitle

\section{Introduction}\label{section:introduction}
\IEEEPARstart{R}{E}{C}{E}{N}{T} advances in mobile computing have promoted transmission of analog-valued signals over wireless fading channels. Examples include multimedia delivery to mobile users and measurements accumulation at the sensor fusion center. However, it is well known that the conventional digital scheme, based on the Shannon separation principle, cannot provide robustness over a wide range of wireless channel conditions. Specifically, it may suffer from \emph {cliff effect} \cite{shannon1949}: performance degrades drastically when the instantaneous signal-to-noise ratio (SNR) drops beneath the target SNR, due to the quantized bit's sensitivity to noise; while performance could not improve even if SNR increases beyond the target SNR, due to the nonrecoverable quantization errors. On the contrary, analog systems, such as amplitude modulation (AM) schemes, can eliminate cliff effect inherently, and are optimal for a Gaussian vector transmitted over a Gaussian channel with equal source-channel bandwidth \cite{gastpar2003tocode}. However, due to low compression efficiency, the analog system is practically inferior to the digital system, especially for the case of bandwidth mismatch between the source and the channel.

\par
To achieve the balance between efficiency and robustness, hybrid digital-analog (HDA) transmission schemes have been proposed, with the capability to outperform digital and analog systems both in theoretical fields \cite{Mittal2002hybrid}, \cite{xu2016hybrid}, \cite{minero2015a}, \cite{koken2015on}, \cite{chen2014zero}, \cite{rungeler2014design} and for image/video applications \cite{tan2017a}, \cite{zhao2016adaptive}, \cite{he2015structure}, \cite{Yu2014scalable}. In HDA, as reconstructed signals consist of digital and analog components, resources (i.e., power and bandwidth) should be wisely assigned to digital and analog codes to achieve optimal performance in terms of distortion \cite{chen2014zero}, \cite{he2015structure}, \cite{skoglund2006hybrid}, \cite{wang2009hybrid}, \cite{lu2016joint}. The authors in \cite{skoglund2006hybrid} constructed HDA schemes with various power allocation coefficients, and demonstrated the importance of power allocation on system distortion. Two vector quantization (VQ) based HDA schemes with bandwidth mismatch were proposed in \cite{wang2009hybrid}, where power allocation was optimized for asymptotical distortion over additive white Gaussian noise (AWGN) channels. The authors in \cite{lu2016joint} implemented HDA transmission in wireless relay network, and proposed a joint power allocation scheme for both digital-analog power allocation and source-relay power allocation. Apart from power allocation, He {\em et al.} in \cite{he2015structure} also took bandwidth allocation into consideration.

\par
Nonetheless, there still remains challenges for resource allocation in HDA transmission over wireless fading channels. First, in most existing studies, the exact channel state information (CSI) is assumed to be available at the transmitter and enough resources are allocated for digital codes to enable error-free digital transmission. Consequently, distortion optimization is only performed by the analog part. However, in wireless fading channels where the exact CSI is not known at the transmitter, such kind of resource allocation strategies cannot be adopted. In fact, due to inevitable random outage, digital distortion cannot be neglected in this case. Second, {rate allocation is important in a general HDA system, where source is split into quantized bits for digital transmission and quantization errors for analog transmission \cite{skoglund2006hybrid}, \cite{wang2009hybrid}. The reason is that rate allocation determines the amount of {information} for digital transmission and that for analog transmission. However, it is preset empirically before resource allocation in early studies.} This operation is much reasonable with perfect CSI assumption, but not valid in fading channels. With consideration of digital distortion, rate allocation and resource allocation are coupled and both have significant impacts on system distortion. More information with higher coding rate or less resources allocated to digital transmission would increase distortion caused by digital outage, at the same time, distortion introduced by analog decoding would be alleviated.

\par
To overcome aforementioned challenges, in this paper, we jointly perform rate and resource allocation in HDA transmission over wireless fading channels, with the goal to minimize expected system distortion. Both analog and digital distortion are considered due to inevitable random outage in wireless fading channels. We first study the joint allocation problem under the single Gaussian source scenario. Further, we extend our work to the multivariate Gaussian (or parallel Gaussian) source scenario, where each Gaussian vector is independent and non-identically distributed with diverse variances from other vectors. In fact, the multivariate Gaussian model has been widely used in many signal processing applications, including sensor data \cite{deshpande2004model}, \cite{liu2012energy}, multimedia data \cite{cui2014robust}, watermarking application \cite{moulin2004parallel}, and object clustering \cite{davis2007differential}. The main contributions are summarized as follows.
\par
1) In an HDA system with the entropy-constrained scalar quantizer (ECSQ)\cite{gish1968asymptotically}, where quantized bits are transmitted in the digital branch and quantization errors are transmitted in the analog branch \cite{skoglund2006hybrid}, \cite{wang2009hybrid}, the quantization rate is configured to determine information split between digital and analog parts. Considering such system with bandwidth expansion over a quasi-static Rayleigh fading channel, orthogonal transmission for analog and digital signals are adopted. We analyze analog distortion as well as digital distortion theoretically under random outage. In the limit of low outage probability, we derive the closed-form expression of expected distortion of the HDA system, based on the asymptotical characteristic of the ECSQ.
\par
2) Under the single Gaussian source scenario, the derived joint rate and resource allocation problem is a non-convex nonlinear problem when the ECSQ is involved. we decompose it into two coupled sub-problems based on the block coordinate descent (BCD) method \cite{bersekas1999nonlinear}. Moreover, we derive the expression of the solution for each individual sub-problem. Based on our theoretical work, we propose an iterative gradient algorithm to approach the optimal solution for the derived problem.
\par
3) Under the multivariate Gaussian source scenario, we solve the joint rate and resource allocation problem in two stages when the ECSQ is considered, i.e. intra-component optimization and inter-component optimization. The former distributedly optimizes rate and resource allocation for each Gaussian vector component, which could be solved analogously as that in the single Gaussian source scenario. The latter attempts to allocate resources among components, which is a mixed integer nonlinear programming (MINLP) problem. An efficient algorithm adopting rounding and greedy strategies is proposed, which can converge to the optimal solution for the slack problem due to its proved convexity. Then a locally optimal integer solution is found in a greedy manner.
\par
The rest of the paper is organized as follows. Related work is discussed briefly in Section \ref{section:related work}. In Section \ref{section:systemModelAssumptions}, the system model of HDA transmission with bandwidth expansion is described in detail. In Section \ref{section:single gaussian source}, we formulate the joint rate and resource allocation problem under single Gaussian source scenario and propose an algorithm to solve it. In Section \ref{section:multi-gaussian source}, we extend our work to the multivariate Gaussian source scenario. Performance evaluation is presented in Section \ref{section:simulation}. Finally, we conclude our work with a summary in Section \ref{section:conclusion}.
\par
{\em Notations:}
Vectors and matrices are denoted by bold-faced characters. Upper-case letters are used for random entities and lower-case letters for their realizations. $[x]^+$ denotes $\max(x,0)$. Unless specified otherwise, the base of logarithms is two. $\mathbb{R}$ and $\mathbb{N}$ represent the set of real numbers and the set of non-negative integers, respectively.

\begin{figure*}[htb]
\centering
\includegraphics[width=12cm]{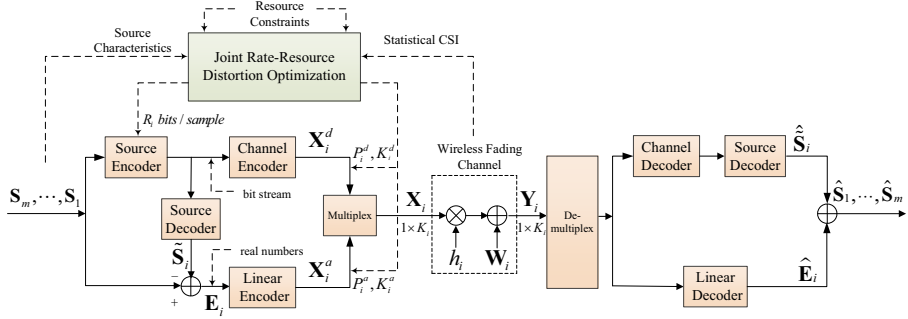}
\caption{System overview of HDA transmission over the wireless fading channel with bandwidth expansion.}\label{figure:systemOverview}
\end{figure*}

\section{Related Work}\label{section:related work}
Due to severely limited resources (i.e., bandwidth and power) and time-varying channel conditions in wireless networks, it is desirable to construct a robust system over a wide range of channel conditions. Recently, HDA transmission, which integrates both analog coding approach and digital coding approach, has attracted wide attention \cite{Mittal2002hybrid}, \cite{xu2016hybrid}, \cite{minero2015a}, \cite{zhao2016adaptive}, \cite{wang2009hybrid}, \cite{lu2016joint}, \cite{pezeshkpour2014optimal}. It has been proven that graceful degradation and robustness over varying SNR can be achieved in HDA systems for Gaussian channels.
\par
System performance of HDA is usually measured in terms of distortion. Several researches focused on theoretically analyzing distortion performance of HDA schemes. Different from separate source and
channel coding, HDA transmission can be categorized into joint source-channel coding (JSCC) schemes, which can achieve noticeable gains in the aspect of distortion performance and coding efficiency \cite{lim2003joint}, \cite{shahidi2013map}, \cite{wu2009design}. The authors in \cite{pezeshkpour2014optimal} studied HDA JSCC schemes of transmitting the Gaussian source over the Gaussian channel with the presence of an interference, where the tradeoff of distortion between the source and the channel state was analyzed. The theoretical analysis characterized the achievable region of such distortion pair. A generalized HDA framework was proposed in \cite{minero2015a} for Gaussian systems. Distortion performance of this unified HDA design was investigated for two typical scenarios: lossy JSCC over multiple access channels and channel coding over relay networks. In above studies, Gaussian channels were assumed. G.Caire and K.Narayananin in \cite{caire2007distortion} considered the Rayleigh fading channel and analyzed the distortion SNR exponent versus spectral efficiency in the limit of high SNR. A tight lower bound of the exponent was shown achievable by a proposed HDA space-time coding scheme.
\par

Resource allocation especially power allocation has been investigated to achieve desired distortion performance in HDA systems \cite{chen2014zero}, \cite{skoglund2006hybrid}, \cite{wang2009hybrid}, \cite{lu2016joint}, \cite{wu2017scalable}. In most existing HDA systems, power allocation is performed with the assumption that perfect CSI is available at the transmitter. Chen {\em et al.} in \cite{chen2014zero} developed an HDA joint source-channel coding method for Wyner-Ziv problem, and further applied such HDA scheme in scenarios without side information but with 1:2 bandwidth expansion. The source utilized one channel use for analog transmission and another channel use for superposed digital and analog transmission. Power allocation of the two channel uses, as well as digital-analog power allocation of the second channel use, were numerically optimized. In \cite{lu2016joint}, HDA transmission was implemented into the wireless relay network, where a joint power allocation scheme has been proposed for both digital-analog power allocation and source-relay power allocation. Wang {\em et al.} in \cite{wang2009hybrid} investigated the effect of digital-analog power allocation on designing the optimal HDA system. Two HDA systems have been proposed for transmission of a Gaussian vector over an AWGN channel under bandwidth compression, where the upper bound on asymptotically optimal distortion was analyzed. Moreover, the authors in \cite{wu2017scalable} implemented HDA transmission in heterogeneous cellular networks, where a femto user received its superposed digital and analog signals from macro base stations or femto access points according to a load balancing parameter. Power allocation for HDA transmission in \cite{wu2017scalable} was analyzed with the stochastic geometry theory.
\par
In an HDA system, the first step is to split the source for analog transmission and digital transmission, which is referred as rate allocation in this paper. Like resource allocation, rate allocation also has significant impact on performance of HDA transmission. As quantization errors are transmitted in the analog branch of the HDA system, rate allocation can be configured by the quantizer. Some work has been carried out to investigate the effect of quantizers in HDA \cite{Coward2000quantizer},\cite{kleiner2009asymptotically}. The hybrid scalar quantizer linear coder was proposed in \cite{Coward2000quantizer} for bandwidth expansion of ratio two. Subsequently, Kleiner {\em et al.} in \cite{kleiner2009asymptotically} extended the system into arbitrary bandwidth expansion ratio, which quantized the source along with repeatedly quantizing the error caused from the previous step, and ended the last step with linear coders. Unfortunately, in these studies, resource allocation has not been considered. A recent study in \cite{tan2017an} has taken both rate allocation and power allocation into consideration for HDA transmission of video. In this literature, rate allocation was controlled by the quantization parameter (QP) and a prediction model was proposed to characterize the relationship between the QP and the data variance of analog data. Based on the model, it optimized power allocation and the QP. However, it only considered transmission under the Gaussian channel with matched bandwidth, its solution could not be directly applied into fading channels without CSI and the scenario that source bandwidth and channel bandwidth are mismatched.

\section{System Model}\label{section:systemModelAssumptions}

The block diagram of the HDA transmission system over the fading channel with bandwidth expansion is shown in Fig.\ref{figure:systemOverview}. The upper part of the system is the digital coder and the lower part is the analog coder. The source is split into quantized bits and quantization errors for digital and analog transmission, respectively. For convenience, the main notations are summarized in Table \ref{tab:notation}.
\subsection{Overview}\label{subsection:overview}
At the transmitter, the time-discrete and analog-valued memoryless vectors are grouped as an $m\times L$ matrix
$\bm{S}=[\bm{S}_1, \bm{S}_2, \cdots, \bm{S}_m]^T$, where $T$ denotes transposition. Each vector is with $L \times 1$ dimension, i.e., $\bm{S}_i=[s_{i1},s_{i2},...s_{iL}]^T$, $\forall i \in \{1, ..., m\}$. Besides, the $L$ samples in each vector (e.g., $\bm S_i$) are independent identically distributed (i.i.d.). At the digital encoder, these vectors are successively fed to the source encoder and channel encoder. For the $i^{th}$ vector $\bm S_i$, its digital codeword is $\bm{X}^d_i$. At the analog encoder, the output of the digital source encoder is first sent to a source decoder, to generate time-discrete and discrete-valued vector $\tilde{\bm{S}_i}$. Then by subtracting $\tilde{\bm{S}_i}$ from $\bm{S}_i$, the quantization error $\bm{E}_i=\bm{S}_i-\tilde{\bm{S}}_i$ can be obtained, which is further coded with a linear encoder into the analog codeword $\bm{X}^a_i$. Finally, $\bm{X}^d_i$ and $\bm{X}^a_i$, which are separatively transmitted with corresponding channel uses, form the transmitted codeword $\bm{X}_i$. Note that the dimension of $\bm{X}_i$ is related to the number of of channel uses assigned to it. Since bandwidth expansion is considered, analog signals and digital signals occupies different numbers of channel uses for orthogonal transmission.
\par
\begin{table}[tb]
\centering
\caption{Parameter Descriptions}\label{tab:notation}
\begin{tabular}{|L{1cm}|L{7cm}|}
     \hline
     {\bf Notation} &  {\bf Description} \\
     \hline
     $m$ & number of Gaussian vectors in multivariate Gaussian source\\
     \hline
     $L$ & number of samples in each Gaussian vector\\
     \hline
     $\bm{S}$ &  multivariate Gaussian vectors matrix, $\bm{S}=[\bm{S}_1, \bm{S}_2, \cdots, \bm{S}_m]^T$\\
     \hline
     $\bm{S}_i$ &  vector whose transposition represents the $i^{th}$ row of $\bm S$, $\bm{S}_i=[s_{i1},s_{i2},...s_{iL}]^T$\\
     \hline
     $P,K$ & power budget, number of available channel uses\\
     \hline
     $\bm{E_i}$ & quantization error vector of $\bm{S}_i$, $L\times1$ vector\\
     \hline
     $\bm{X}_i$ & channel codeword of $\bm{S}_i$, $K_i\times1$ vector\\
     \hline
     $P_i,K_i$ & power, number of channel users assigned for $\bm{X}_i$\\
     \hline
     $\bm{X}_i^a$ & analog channel codeword of $\bm{S}_i$, $K_i^a\times1$ vector\\
     \hline
     $P_i^a,K_i^a$ & power, number of channel users assigned for $\bm{X}_i^a$\\
     \hline
     $\bm{X}_i^d$ & digital channel codeword of $\bm{S}_i$, $K_i^d\times1$ vector\\
     \hline
     $P_i^d,K_i^d$ & power, number of channel users assigned for $\bm{X}_i^d$\\
     \hline
     $\bm{W_i}$ & channel noise during transmission of $\bm{X}_i$, $K_i\times1$ vector\\
     \hline
     $\bm{Y}_i$ & received signal corresponding to $\bm{X}_i$, $K_i\times1$ vector\\
     \hline
     $\sigma_i^2$ & variance of $\bm{S}_i$\\
     \hline
     ${\bm \Sigma_S}$ & correlation matrix of $\bm S$ with $m\times m$ dimension\\
     \hline
     $R_i$ & quantization rate of $\bm{S}_i$\\
     \hline
\end{tabular}
\end{table}
At the receiver, the decoder maps the channel output $\bm{Y}_i$ into the estimation $\widehat{\bm{S}}_i$. Ideally, the mapping should be selected to minimize the mean squared-error (MSE) distortion $E||{\bm{S}_i}-\widehat{\bm{S}}_i||^2$. However, the high implementation complexity prohibits the use of such a decoder. To enable practical HDA transmission, we adopt a suboptimal decoder proposed in \cite{skoglund2006hybrid}, as shown in Fig. \ref{figure:systemOverview}. The received digital and analog signals are separately decoded. The decoding process of the digital signal is inverse of the digital encoding process, while the analog decoder is a linear minimum mean square error (MMSE) estimator. The add of $\widehat{\tilde{\bm{S}}}_i$ and $\widehat{\bm{E}}_i$, which are corresponding outputs of the digital and analog decoders, respectively, forms the reconstructed vector $\widehat{\bm{S}}_i$.

\subsection{Source and Channel}\label{subsection:source and channel}

We assume $\bm{S}=[\bm{S}_1, \bm{S}_2, \cdots, \bm{S}_m]^T$ to be a zero-mean multivariate Gaussian complex vector. $\bm S$ is said to be multivariate Gaussian if it follows Gaussian distribution $\mathcal{N}({\bm \mu_S},{\bm \Sigma_S})$, where $\bm \mu_S$ is an $m$-length vector of means, and ${\bm \Sigma_S}$ is an $m \times m$ correlation matrix. Throughout the paper, $\bm \mu_S$ is a all-zero vector $\bm 0$, and ${\bm \Sigma_S}={\rm diag} \, \{\sigma_i^2\}_{i=1}^{i=m}$ is considered to be diagonal, where $\sigma_i^2$ is the variance of the $i^{th}$ vector $\bm S_i$. This implies that the vectors are independent with each other. If the correlation matrix ${\bm \Sigma_S}$ is not diagonal, $\bm{S}$ can be de-correlated by multiplying the Karhunen-Loeve transforming matrix $\bm T$, i.e., $\bm {TS}$. Such diagonal assumption and de-correlation for signals can be traced in early work \cite{skoglund2006hybrid}, \cite{wang2009hybrid}, \cite{cui2014robust}. In fact, such multivariate Gaussian source model is practical to characterize wide data analysis applications \cite{deshpande2004model}, \cite{liu2012energy}, \cite{cui2014robust}, \cite{moulin2004parallel}, \cite{davis2007differential}. Although $L$ samples in each vector is complex Gaussian i.i.d. (i.e., $\bm{S}_i\sim\mathcal{CN}(0,\sigma_i^2), ~\forall i$), samples of different vectors are non-identically distributed with different variances. Without a loss of generality, all vectors in $\bm S$ are in the descending order of their variances, i.e., $\sigma_1^2\geq\sigma_2^2\geq...\geq\sigma_m^2$.
\par

Assume that $K$ uncorrelated channel uses are available, and the channel coherence time is much larger than the sampling time of each Gaussian vector. Based on this assumption, the channel can be modeled as a quasi-static Rayleigh fading channel. Specifically, the channel is assumed to be constant over the duration of any vector transmission but random with Rayleigh fading coefficient $h_i\sim\mathcal{CN}(0,1),\forall i$. Thus, the channel is assumed as non-ergodic, which implies that the exact $h_i$ is not available at the transmitter. $\bm W_i$ is complex additive white Gaussian noise, corresponding to transmission of $\bm{X}_i$. Hence, such channel can be modeled as
\begin{eqnarray}\label{equation:digitalchannel}
\bm{Y_i}=h_i\bm{X}_i+\bm{W}_i, ~~~\forall i.
\end{eqnarray}
\par
The power of Gaussian noise is assumed to be same among all available channels, denoted as $\sigma_w^2$. The system average SNR per channel use is
\begin{eqnarray}\label{equation:generalSNR}
\gamma=\frac{P}{K\sigma_w^2},
\end{eqnarray}
where $P$ is the system power budget.
\subsection{{\ Rate} Allocation}\label{subsection:data allocation}
Some theoretical and practical methods are available for {information split between digital and analog parts} in HDA. For example, dimension splitting is adopted in \cite{wang2009hybrid}, where an $n$-dimensional vector is split into the $n'$-dimensional digital vector and $(n-n')$-dimensional analog vector. For video transforming coefficients, more important coefficients are conveyed by digital transmission and the residual is conveyed by analog transmission \cite{zhao2016adaptive}, \cite{he2015structure}. In this research, we adopt a general HDA framework, where the source is split into quantized bits for digital transmission and quantization errors for analog transmission. In this case, digital-analog information split can be configured with the quantization rate, which is called rate allocation. This approach to achieve rate allocation for HDA transmission is common in theoretical fields \cite{rungeler2014design}, \cite{lu2016joint} and practical fields \cite{tan2017a}.
\par
To enable quantization with variable rates, the $i^{th}$ vector $\bm S_i$ is assumed to be encoded by the ECSQ, with rate $R_i$ bits per sample. With the ECSQ, $R_i$ is actually the entropy of the encoder output. In fact, it has been shown that the quantization error $\bm{E}_i$ is asymptotically Gaussian distributed with zero mean and variance $\sigma_{ei}^2$ \cite{wang2009hybrid}. Furthermore, for the ECSQ adopted in this research, Gish and Pierce \cite{gish1968asymptotically} have demonstrated that the rate of the quantizer asymptotically exceeds the rate given by Shannon rate-distortion function, with a constant parameter $\frac{1}{2}\log\frac{\pi e}{6}$, i.e.,
\begin{equation}\label{equation:rate-distortion}
R_i(\sigma_{ei}^2)=\frac{1}{2}\log{\frac{\sigma_i^2}{\sigma_{ei}^2}}+\frac{1}{2}\log\frac{\pi e}{6}.
\end{equation}
This also indicates that the analog signal is determined by the quantization rate.
\subsection{Resource Allocation}\label{subsection:resource allocation}
Let $\bm{K}=(K_1,K_2,...,K_m),\bm{P}=(P_1,P_2,...,P_m)$ respectively be the channel and power allocation vetors among multivariate Gaussian components. For arbitrary $1\leq i \leq m$, the digital and analog coded signals of the $i^{th}$ Gaussian vector $\bm S_i$ should respectively satisfy the power constraint as
\begin{eqnarray*}\label{equation:generalSNR}
E||{\bm{X}^d_i}||^2\leq \alpha_i P_i, \quad E||{\bm{X}^a_i}||^2\leq (1-\alpha_i) P_i,
\end{eqnarray*}
where $\alpha_i \in [0,1]$ is the power allocation coefficient for digital transmission of ${\bm S}_i$. Furthermore, the average SNR per channel use for digital transmission and that for analog transmission of $\bm S_i$ is:
\begin{equation}\label{equation:digitalAnalogSNR}
\gamma^d_i=\frac{\alpha_i P_i}{K^d_i {\  \sigma_w^2}}, \quad \gamma^a_i=\frac{(1-\alpha_i)P_i}{K^a_i{\  \sigma_w^2}}.
\end{equation}
where $K^d_i$ and $K^a_i$ are the number of channel uses (bandwidth) for digital transmission and that for analog transmission, respectively. In this research, linear AM is adopted as the analog encoder. The number of of channel uses occupied by analog transmission is set to be equal to the number of source samples, i.e., $K^a_i=L$, $K^d_i=K_i-L$. Note that, in the case of bandwidth expansion, each vector is transmitted, which implies that the available channel uses are more than samples. Thus, $K_i > L$ for arbitrary $i\in(1,2,\cdots,m)$ and $\Sigma_{i=1}^mK_i > mL$. Besides, system transmission should meet the bandwidth budget, namely, $\sum\nolimits_{i=1}^m K_i \leq K$.
\par
\section{Joint Rate and Resource Allocation for Single Gaussian Source Scenario}\label{section:single gaussian source}
In this section, we jointly optimize rate allocation and resource allocation for single complex Gaussian vector HDA transmission with bandwidth expansion. Given fading statistics and specific codes, the expression of expected system distortion is first derived with consideration of digital distortion and analog distortion. Then we formulate the joint optimization problem as a minimization problem on expected system distortion, in the limit of low outage probability and with ECSQ. We finally propose an iterative algorithm adopting BCD and gradient methods, based on the analytical efforts.

\subsection{Problem Formulation}\label{subsection:Problem Formulation}
Without a loss of generality, we take transmission of the $i^{th}$ vector $\bm S_i$ into consideration in this scenario.
\subsubsection{Outage Formulation}
The commonly adopted concept for non-ergodic fading channels is information outage \cite{ozarow1994information}, \cite{li2016outage}. It can be defined as the event that the mutual information of the channel could not support a certain data rate \cite{ozarow1994information}. Mathematically, the information outage event can be expressed as
\begin{eqnarray*}\label{equation:digitalCapacity}
\mathscr{A}_i=\{h_i:\mathcal{I}(\bm{X}_i;\bm{Y}_i|h_i) < R_{ti}\},
\end{eqnarray*}
where the data rate is $R_{i}^t=LR_i/K^d_i$ bits per channel use.
The mutual information of the non-ergodic channel characterized in (\ref{equation:digitalchannel}) can be expressed as \cite{telatar1999capacity}
\begin{eqnarray}\label{equation:mutual-information}
\begin{split}
\mathcal{I}(\bm{X}_i;\bm{Y}_i|h_i)\!=\!\mathcal{H}(\bm{Y}_i)-\mathcal{H}(\bm{Y}_i|\bm{X}_i)\!=\!\log\,(1+{\  |h_i|^2}\gamma^d_i).
\end{split}
\end{eqnarray}
\par
Hence, the outage probability can be derived as
\begin{equation}\label{equation:outageProbability-original}
\begin{split}
P^{out}_i=P\{{\rm log} \, (1+{\  |h_i|^2}\gamma^d_i)<R_{ti}\}=F_{g_i}(\frac{2^{R_{ti}}-1}{\gamma^d_i}),
\end{split}
\end{equation}
where $g_i=|h_i|^2$ is the channel fading gain and $F_{g_i}(\cdot)$ is the cumulative distribution function (CDF) of $g_i$. Since $g_i$ is exponentially distributed with parameter $\lambda=1$, the outage probability can be reexpressed as
\begin{equation}\label{equation:outageProbability}
P^{out}_i = 1-e^{-\frac{2^{R_{i}^t}-1}{\gamma^d_i}}.
\end{equation}

\subsubsection{Expected System Distortion}
Distortion between the vector and its reproduction is a common parameter to characterize performance of transmission. With the criteria of MSE, system distortion is
\begin{equation*}\label{equation:quadratic-distortion}
\begin{split}
D_i\triangleq E||{\bm{S}_i}-\widehat{\bm{S}}_i ||^2=\frac{1}{L}\sum_{l=1}^{L}(s_{il}-\widehat{s}_{il})^2.
\end{split}
\end{equation*}
Given channel distribution, expected system distortion $E[D_i]$ is the expectation of $D_i$. Throughout this paper, the expectation is taken with respect to the random fading gain $g_i$. For expression simplicity, we denote expected system distortion as $ED_i$ in the subsequent discussion.
\par
As outage happens randomly, system distortion is caused by both digital transmission and analog transmission. As Fig. \ref{figure:systemOverview} shows, the reconstructed signal at the receiver is $\widehat{\bm{S}}_i=\widehat{\tilde{\bm{S}}}_i+\widehat{\bm{E}}_i$.
When outage happens in digital transmission, the digital signal cannot be decoded at all. In such case, $\widehat{\bm{S}}_i=0$ and system distortion is the variance of the Gaussian vector, i.e., $\sigma_i^2$. On the other hand, if outage does not happen, the digital signal can always be decoded correctly based on the assumption of ideal coding scheme. In this case, distortion of the digital part is $E||\tilde{\bm{S}}_i-\widehat{\tilde{\bm{S}}}_i||^2=0$. Moreover, it has been shown that the quantization error $\bm{E}_i$ can be considered as asymptotically uncorrelated with $\tilde{\bm{S}}_i$ \cite{Mittal2002hybrid}. Thus we have
\begin{equation*}\label{equation:distortionExpression}
\begin{aligned}
{E}||\bm{S}_i-\widehat{\bm{S}}_i||^2&={E}||({\tilde{\bm{S}}_i}+\bm{E}_i)-(\widehat{\tilde{\bm{S}}}_i+\widehat{\bm{E}}_i)||^2\\
&=E||\bm{E}_i-\widehat{\bm{E}}_i||^2.
\end{aligned}
\end{equation*}
Therefore, system distortion when outage does not occur, is merely the linear MMSE estimation error of the analog signal, i.e., $mmse(\bm{E}_i|g_i)$. Based on the above analysis, expected system distortion can be expressed as:
\begin{equation}\label{equation:distortionExpression1}
\begin{split}
ED_i&={ED_i}_{\mathscr{A}_i}+{ED_i}_{\mathscr{A}_i^c}\\
 &=\int_{\mathscr{A}_i} \sigma_i^2\,dF_{g_i}+\int_{\mathscr{A}_i^c} mmse(\bm{E}_i|g_i)\,dF_{g_i},
\end{split}
\end{equation}
where the integral is taken over the fading gain $g_i$. $\mathscr{A}_i^c$ represents the event of non-outage. Based on (\ref{equation:rate-distortion}), the ratio of $\sigma_{ei}^2$ and $\sigma_i^2$ is
$\sigma_{ei}^2 / \sigma_i^2 = \frac{\pi e}{6}2^{-2R_i}$. Hence, the variance of quantization errors decreases rapidly with increasing $R_i$. We introduce a large enough threshold $R_i^{th}$, based on the results of \cite{netravali1976optimum}. When $R_i$ is above a threshold $R_i^{th}$, $\sigma_{ei}^2 / \sigma_i^2$ is small enough, and the quantization is fine enough. In this case, quantization errors contribute little to the performance improvement. In other word, purely digital transmission, not HDA transmission, should be considered when $R_i$ is above a certain threshold. Based on the above discussion, the joint optimization problem for HDA transmission over fading channels can be formulated as a constrained expected system distortion minimization problem, which is expressed as
\begin{equation}\label{equation:single-optimization1}
\begin{aligned}
\!\mathcal{P}1:~
&\min_{\alpha_i,R_i,K_i^d}
&& ED_i={ED_i}_{\mathscr{A}_i}+{ED_i}_{\mathscr{A}_i^c}\\
& ~~~~\text{s.t.}
&& {0 \leq R_i \leq R_i^{th}, ~~0\leq \alpha_i \leq 1,}\\
&&& K_i^d + K_i^a = K_i,\\
&&& R_i,\alpha_i\in\mathbb{R}, ~~K_i^d\in\mathbb{N}.
\end{aligned}
\end{equation}
\subsection{Problem Analysis}\label{subsection:Problem Analysis}
We first analyze expected system distortion in the case that outage happens. Based on (\ref{equation:outageProbability}), it can be rewritten and approximated as
\begin{equation}\label{equation:outage-distortion}
\begin{split}
{ED_i}_{\mathscr{A}_i}&=\int_{\mathscr{A}_i} \sigma_i^2\,dF_{g_i}=P_i^{out}\sigma_i^2\\
&=(1-e^{-\frac{2^{R_{i}^t}-1}{\gamma_i^d}})\sigma_i^2 \doteq \frac{2^{R_{i}^t}-1}{\gamma_i^d}\sigma_i^2
\end{split}
\end{equation}
\par
Define $\tau_i=\frac{2^{R_{i}^t}-1}{\gamma_i^d}$. Then the above approximation holds with $\tau_i\rightarrow0$, as $P_i^{out}=(1-e^{-\tau_i})\doteq\tau_i$. From this, we can find that $\tau_i\rightarrow0$ implies the limit of low outage probability. For this general HDA scheme, quantized bits are transmitted via digital codes. While quantization errors are transmitted via analog codes, act as the supplement to enhance performance. In this case, once outage happens, both digital and analog signals cannot be decoded. Thus, performance of HDA transmission would degrade rapidly with increasing outage probability. And in the following content, we mainly analyze the problem in the limit of low outage probability, i.e., $P_i^{out}\rightarrow0$. The simulation result of Fig. \ref{figure:single_mesh} also shows that the solution obtained with approximation can also achieve comparable performance in practical implementation. Similarly approximation can be referred to  \cite{caire2007distortion}, \cite{li2016outage}, \cite{gunduz2005source}.
\par
Next we begin to analyze expected distortion when outage does not happen. As analyzed in \ref{subsection:data allocation}, the quantization error $\bm{E}_i$ is asymptotically Gaussian distributed. With the linear MMSE, the estimation error of $\bm{E}_i$ is
\begin{equation}\label{equation:mmse-distortion}
E||\bm{E}_i-\widehat{\bm{E}}_i||^2=\frac{\sigma_{ei}^2}{1+g_i\gamma_i^a},
\end{equation}
where $\sigma_{ei}^2$ is the variance of $\bm{E}_i$. The relationship between $R_i$ and $\sigma_{ei}^2$ is expressed in (\ref{equation:rate-distortion}).
\par
For expression brevity, define
\begin{equation}\label{equation:Psi-function}
\Psi(x)= \int_{0}^{\infty} \frac{1}{1+g_i x}p(g_i)\,dg_i.
\end{equation}
According to (\ref{equation:rate-distortion}), (\ref{equation:mmse-distortion}) and (\ref{equation:Psi-function}), expected system distortion with absence of outage can be rewritten as
\begin{equation}\label{equation:non-outage-distortion}
\begin{split}
{ED_i}_{\mathscr{A}_i^c}&=\int_{\mathscr{A}_i^c} mmse(\bm{E}_i|g_i)\,dF_{g_i}\\
                 &=E\{\frac{\sigma_{ei}^2}{1+g_i\gamma_i^a}|\log(1+g_i\gamma_i^d)\geq R_{ti}\}\\
                 &=\sigma_{ei}^2\int_{\tau_i}^{\infty} \frac{1}{1+g_i\gamma_i^a}p(g_i)\,dg_i\\
                 &\doteq \frac{\pi e}{6}2^{-2R_i}\sigma_i^2\Psi(\gamma_i^a),
\end{split}
\end{equation}
where $p(g_i)=e^{-g_i}$ is the probability density function (pdf) of channel gain $g_i$.
The last line of (\ref{equation:non-outage-distortion}) holds with small $\tau_i$ approximation.
\par
Substituting (\ref{equation:outage-distortion}) and (\ref{equation:non-outage-distortion}) into (\ref{equation:single-optimization1}), the problem $\mathcal{P}1$ can be reformulated approximatively as follows.
\begin{equation}\label{equation:single-optimization2}
\begin{aligned}
\mathcal{P}2:~
&\min_{\alpha_i,R_i,K_i^d}
&& ED_i=\frac{2^{R_{i}^t}-1}{\gamma_i^d}\sigma_i^2+\frac{\pi e}{6}2^{-2R_i}\sigma_i^2\Psi(\gamma_i^a)\\
& ~~\text{s.t.}
&& 0 \leq R_i \leq R_i^{th}, ~~0\leq \alpha_i \leq 1,\\
&&& K_i^d + K_i^a = K_i,\\
&&& \!R_i,\alpha_i\in\mathbb{R}, ~~K_i^d\in\mathbb{N},
\end{aligned}
\end{equation}
where $\sigma_i^2$ is the variance of $\bm S_i$, $\gamma_i^d$ and $\gamma_i^a$ are the average digital SNR and analog SNR, respectively, as expressed in (\ref{equation:digitalAnalogSNR}).
\par
Such coupled optimization objective is a constrained multi-variable function, which contains integral components. It is difficult to derive the optimal solution with general convex optimization methods. Inspired by the BCD method \cite{bersekas1999nonlinear}, we tackle such problem by dividing it into two coupled sub-problems: optimizing rate allocation under given resource allocation, and then optimizing resource allocation based on derived rate allocation.
\subsubsection{Rate Allocation}\label{subsubsection:Data Segmentation}
Under given resource allocation, the rate allocation problem is a relaxed problem of (\ref{equation:single-optimization2}), which can be expressed as
\begin{equation}\label{equation:single-optimization-sub1}
\begin{aligned}
&\min_{R_i}
&& ED_i(R_i|\alpha_i,K_i^d)\\
&~ \text{s.t.}
&& {\  0 \leq R_i \leq R_i^{th}}, R_i\in\mathbb{R}.
\end{aligned}
\end{equation}
\par
For this sub-problem, we can derive its analytical solution as follows.
\begin{theorem}
Let ($\alpha_i,K_i^d$) be a given digital-analog resource allocation strategy. Then the quantization rate $R_i^{*}$ is the optimal solution for the problem formulated in (\ref{equation:single-optimization-sub1}) if and only if
\begin{equation}\label{equation:rate-condition}
R_i^{*}=
    min([\frac{K_i^d}{K_i+K_i^d}{\rm log} \,(\frac{\pi e}{3}\frac{P_i\alpha_i}{{\  \sigma_w^2}L}\Psi(\gamma_i^a))]^+,~R_i^{th}).
\end{equation}
\end{theorem}
\begin{proof}
The second derivative of function $ED_i(R_i|\alpha_i,K_i^d)$ with respect to $R_i$ is $\ln^2 2\frac{L^2}{(K_i^d)^2}\frac{1}{\gamma_i^d}2^{R_{i}^t}\sigma_i^2+\frac{2 \pi e}{3}\ln^2 2\cdot2^{-2R_i}\sigma^2\Psi(\gamma_i^a)$, which is positive. Thus, such sub-problem is strictly convex, which has unique solution. This implies that this function can be minimized by finding the stationary point.
\par
Then we analyze the existence of the optimal solution. The first derivative of $ED_i(R_i|\alpha_i,K_i^d)$ with respect to $R_i$ can be derived as $\ln 2\frac{L}{K_i^d}\frac{1}{\gamma_i^d}2^{R_{i}^t}\sigma_i^2-\frac{\pi e}{3}\ln 2\cdot2^{-2R_i}\sigma_i^2\Psi(\gamma_i^a)$.
By equating the derivative to zero, we can derive the stationary point. And considering the constraint of rate, the optimal solution can be obtained as expressed in (\ref{equation:rate-condition}).
\end{proof}

\begin{remark}
According to Theorem 1, the optimal quantization rate $R_i$ might be zero. In this case, no information is allocated for digital transmission, which implies that the optimal scheme reduces to the purely analog transmission scheme.
\end{remark}
\par
For such a nonlinear optimization problem, the above analysis indicates that the rate allocation problem can be solved efficiently according to (\ref{equation:rate-condition}).

\subsubsection{Resource Allocation}\label{subsubsection:Resource Allocation}
Based on the solution of rate allocation, the problem of resource allocation can be written as
\begin{equation}\label{equation:single-optimization-sub2}
\begin{aligned}
&\underset{\alpha_i, K_i^d}{\text{min}}
&& ED_i(\alpha_i,K_i^d|R_i)\\
&~ \text{s.t.}
&& {\  0\leq\alpha_i\leq1},~~\alpha_i\in\mathbb{R}, \\
&&&~K_i^d + K_i^a = K_i, ~~K_i^d \in\mathbb{N}.
\end{aligned}
\end{equation}
\par
As assumed in \ref{subsection:resource allocation}, analog transmission occupies $K_i^a=L$ channel uses. Therefore, bandwidth allocation is determined by the dimension of the Gaussian vector, i.e., $K_i^d=K_i-L$. Next, we focus on deriving the necessary condition for the optimal power allocation coefficient via the analysis of the Karush-Kuhn-Tucker (KKT) condition \cite{Boyd2014convex}.

\begin{theorem}
Given the transmission resource constraints $P_i,K_i$ and the quantization rate $R_i$, let $\alpha_i^{v}$ satisfy
\begin{equation}\label{equation:ratio-stationary}
\frac{\pi e}{6}\frac{(\alpha_i^{v})^2P^2}{(K_i^d)^2 L {\sigma_w^4}}\frac{\partial \Psi (\gamma_i^a)}{\partial \gamma_i^a}+(2^{R_{i}^t}-1)2^{2R_i}=0,
\end{equation}
then the optimal digital-analog power allocation coefficient $\alpha_i^*$ of the problem formulated in (\ref{equation:single-optimization-sub2}) must satisfy
\begin{equation}\label{equation:ratio-condition}
\alpha_i^{*}=1+(\alpha_i^v-1)\cdot sgn([\alpha_i^v]^+).
\end{equation}
$sgn(x)$ is a sign function which equals one with $0\leq x\leq1$ and zero otherwise.
\end{theorem}
\begin{proof}
According to the Leibniz integral rule, the second derivative of $D_i(\alpha_i|R_i,K_i^d)$ with respect to $\alpha_i$ is $2\frac{P_i^2}{L^2 \sigma_w^4}\frac{\pi e}{6}2^{-2R_i}\sigma_i^2\int_{0}^{\infty} \frac{g_i^2}{(1+g_i \gamma_i^a)^3}p(g_i)\,dg_i+2(2^{R_{i}^t}-1)\frac{P_i^2}{(K_i-L)^2\sigma_w^4}\frac{1}{(\gamma_i^d)^3}\sigma_i^2$, which is positive. This indicates that this sub-problem is strictly convex, and also has unique solution.
\par
Next we analyze whether an optimal point exists. The Lagrangian function of the problem formulated in (\ref{equation:single-optimization-sub2}) can be expressed as
\begin{equation}\label{equation:lagrangian-function}
\begin{split}
\mathcal{L}(\alpha_i,\mu_1,\mu_2)=&\frac{2^{R_{i}^t}-1}{\gamma_i^d}\sigma_i^2+\frac{\pi e}{6}2^{-2R_i}\sigma_i^2\Psi(\gamma_i^a)\\
&+\mu_1(-\alpha_i) + \mu_2(\alpha_i-1).
\end{split}
\end{equation}
\par
The optimal solution can be obtained by applying the KKT condition if and only if there exists Lagrange multipliers $\mu_1,\mu_2$ such that

\begin{subequations}
\begin{numcases}{}
\frac{\partial \mathcal{L}(\alpha_i,\mu_1,\mu_2)}{\partial \alpha_i}=f(\alpha_i^*)-\mu_1^*+\mu_2^* =0& \label{KKT_1}\\
\mu_1^*\alpha_i^*=0,\mu_2^*(\alpha_i^*-1)=0&\label{KKT_2}\\
\mu_1^*\geq0,\mu_2^*\geq0&\label{KKT_3}\\
0\leq\alpha_i^*\leq 1,&\label{KKT_4}
\end{numcases}
\end{subequations}
where $f(\alpha_i)$ can be written as
\begin{equation*}\label{equation:f}
\begin{aligned}
f(\alpha_i)=&-\frac{P_i}{L{\  \sigma_w^2}}\frac{\pi e}{6}2^{-2R_i}\sigma^2\frac{\partial \Psi (\gamma_i^a)}{\partial \gamma_i^a}\\
&-(2^{R_{i}^t}-1)\frac{P_i}{(K_i-L){\  \sigma_w^2}}\frac{1}{(\gamma_i^d)^2}\sigma_i^2.
\end{aligned}
\end{equation*}
\par
Substituting (\ref{KKT_1}) into (\ref{KKT_2}), the necessary condition for optimal $\alpha_i^*$ expressed as (\ref{equation:ratio-condition}) holds, based on (\ref{equation:ratio-stationary}).
\end{proof}
\begin{remark}
Observe that the solution $\alpha_i^*=1$ means no power is allocated to analog transmission, which implies the optimal scheme reduces to purely digital transmission. Such situation might occur when the quantization is fine enough so that the quantization residual is negligible to the signal reconstruction.
\end{remark}
\par
The above analysis suggests that the digital-analog power allocation problem can be solved based on (\ref{equation:ratio-condition}). Recognizing that the integral component contains variable $\alpha_i$, the computational complexity of obtaining $\alpha_i^*$ directly from (\ref{equation:ratio-stationary}) is high. However, the analysis also indicates that the formulation (\ref{equation:single-optimization-sub2}) is a constrained convex optimization problem, which can be addressed tractably through the gradient method.

\subsection{Solution}\label{subsection:single-solution}
According to the above analysis, we develop an iterative solution to jointly optimize rate and resource allocation as formulated in (\ref{equation:single-optimization2}). The algorithm is elaborated in detail in Algorithm \ref{alg:Optimal Strategy for Single-Gaussian Source}. Such iterative idea can be traced from the widely used BCD method proposed in \cite{bersekas1999nonlinear}. For the optimization of a multi-variable function in the BCD method, the coordinates of variables are first partitioned into blocks. Then at each iteration, the function is optimized in terms of one of the coordinate blocks while the other coordinates are fixed. With BCD, one can find the optimal solution with an acceptable convergence rate, even if the objective function is not (block) convex \cite{luo1993error}, \cite{grippo2000on} or differentiable \cite{tseng2001convergence}. The following theorem establishes the optimality of Algorithm 1 for the approximative problem $\mathcal{P}1$.

 \begin{algorithm}
  \caption{Joint Allocation Strategy for the Single Gaussian Source Scenario}
   \label{alg:Optimal Strategy for Single-Gaussian Source}
   \KwIn{$\sigma_i^2,~L,~K_i,~P_i,~\sigma_{ei}^2,~\epsilon,~\delta$}
   \KwOut{$\alpha_i^*,~R_i^*$}
   \textbf{initialize procedure:}\\
   $t=0$\;
   ${K_i^a}=L,{K_i^d}=K-L$\;
   $R_i^*=R_i^{(0)},~\alpha_i^*=\alpha_i^{(0)},~ED_i^{(0)}=+\infty$\;
   \SetKw{Kwend}{end initial procedure}
   \Kwend\;

   \Repeat{$||ED_i^{(t)}-ED_i^{(t-1)}|| \le \epsilon$}{
    $t = t+1$\;
    Calculate $R_i^*$ according to (\ref{equation:rate-condition})\;

    \eIf {$R_i^*=0$}{
    $\alpha_i^*=0$, Calculate $ED_i^{(t)}$ according to (\ref{equation:single-optimization2})\;
    }{
    $\tilde{ED}_i^{(0)} = +\infty, j = 0$\;
    \Repeat{$||\tilde{ED}_i^{(j)}-\tilde{ED}_i^{(j-1)}|| \leq \epsilon$}{
    $j = j+1$\;
    $\theta=-\partial ED_i(\alpha_i|R_i^*,K_i^d)/\partial \alpha_i |_{\alpha_i=\alpha_i^*}$\;
    $\alpha_i^*=\min\{[\alpha_i^*+\delta\theta]^+,1\}$\;
    Calculate $\tilde{ED}_i^{(j)}$ according to (\ref{equation:single-optimization2})\;
    }
    $ED_i^{(t)} = \tilde{ED}_i^{(j)}$\;
    }
    }
   \Return $\alpha^*_i,~R^*_i$\;
   \end{algorithm}

\par

\begin{theorem}
Algorithm 1 converges; and the converged solution $(R^*_i, \alpha^*_i)$ is the optimal
solution to the problem in (\ref{equation:single-optimization2}).
\end{theorem}

This theorem can be proved, based on the fact that each sub-problem has unique solution and the constraint set is a Cartesian product of closed convex sets \cite{bersekas1999nonlinear}. The details of proof can be found in \cite{bersekas1999nonlinear}. In Algorithm \ref{alg:Optimal Strategy for Single-Gaussian Source}, we implement the gradient method to solve the problem of power allocation (Lines 13-19). The power allocation coefficient is evolved on the gradient descent direction $\theta$. Note that in the process of rate allocation and power allocation, the expectation of distortion as defined in (\ref{equation:single-optimization2}) needs to be computed repeatedly. To reduce the complexity of calculating integral component $\Psi(\gamma_i^a)$, we further analyze the integral and adopt a lookup table for it.
\par
The fading gain $g_i$ is exponentially distributed with pdf $p(g_i)=e^{-g_i}$ as analyzed in \ref{subsection:Problem Analysis}. Thus the function $\Psi(x)$ can be rewritten as $\Psi(x)=\int_{0}^{\infty} \frac{1}{1+g_ix}e^{-g_i}\,dg_i$.
Further, we can derive it as follows,
\begin{eqnarray*}\label{equation:integral-component}
\begin{split}
\Psi(\frac{1}{x})&=\int_{0}^{\infty} \frac{x}{x+g_i}e^{-g_i}\,d\omega=xe^{x}\int_{x}^{\infty} \frac{1}{g_i}e^{-g_i}\,dg_i\\
&=x(-e^xEi(-x)),
\end{split}
\end{eqnarray*}
where $Ei(-x)$ is the widely known exponential integral function. Therefore, a lookup table can be implemented to calculate the function $\Psi(x)$, which can further reduce the computational complexity.

\section{Joint Rate and Resource Allocation for Multivariate Gaussian Source Scenario}\label{section:multi-gaussian source}
In this section, we further extend our work to the multivariate Gaussian source scenario. In this case, the variance of each Gaussian vector, which characterizes the transmission priority, would accordingly affect both {rate} allocation and resource allocation. Thus, we investigate the joint optimization problem with consideration of not only the fading distribution, but also source characteristics differences. Finally, we propose an efficient algorithm combining rounding and greedy strategies.
\subsection{Problem Formulation}\label{problemformulation}
In this scenario, suppose a set of zero mean complex Gaussian vectors are grouped as $\bm S$ with the correlation matrix ${\bm \Sigma_S}$. And ${\bm \Sigma_S}$ is diagonal with the non-zero entry $\sigma_i^2$, $1\leq i \leq m$. If the vectors are correlated, we can diagonalize the matrix as $\bm{TS}$, where $\bm{T}$ is the Karhunen-Loeve transforming matrix. Without a loss of generality, we suppose $\sigma_1^2\geq\sigma_2^2\geq...\geq\sigma_m^2$. The total power and bandwidth of the HDA system are $P$ and $K$, while the average noise power per channel use is $\sigma_w^2$, as assumed in \ref{subsection:source and channel}. As analyzed in Section \ref{section:single gaussian source}, rate allocation and digital-analog resource allocation of each vector component should be optimized jointly to minimize expected system distortion. Moreover, due to the difference of source characteristics, the optimization of power allocation $\bm{P}=(P_1,...,P_m)$, as well as bandwidth allocation $\bm{K}=(K_1,...,K_m)$ among vectors components, should also be considered. Mathematically, such problem can be formulated as follows.
\begin{equation}\label{equation:muti-optimization1}
\begin{aligned}
\!\!\!\mathcal{P}3:~
&\min_{\bm{P},\bm{K},\bm{R},\bm{\alpha},\bm{K^d}}
&& ED=\sum_{i=1}^mED_i\\
& ~~~~~\text{s.t.}
&& \!\sum_{i=1}^mK_i \leq K, ~~\sum_{i=1}^mP_i \leq P, \\
&&& L < K_i, ~~K_i\in\mathbb{N},\\
&&& 0\leq P_i, ~~P_i\in\mathbb{R},\\
&&& 0\leq R_i\leq R_i^{th}, ~~0\leq \alpha_i \leq 1,\\
&&& K^d_i + K^a_i = K,~~K^d_i \in\mathbb{N},\\
&&& \!R_i,\alpha_i\in\mathbb{R}, ~~\forall i\in \{1,...,m\},
\end{aligned}
\end{equation}
where $ED_i$ is expected system distortion of $\bm S_i$ as (\ref{equation:single-optimization2}) expressed. $K_i$ and $P_i$ are the corresponding channel bandwidth and power, respectively. $R_i$ and ($\alpha_i,K_i^d$) represent rate allocation and digital-analog resource allocation corresponding to $\bm S_i$, respectively, which should satisfy the corresponding constraints.
\subsection{Problem Analysis}
Due to the consideration of source characteristics differences, the problem formulated in (\ref{equation:muti-optimization1}) is hard to be solved. To tame the complexity, we decompose the problem and propose a two-stage algorithm inspired by the perspective of formulated problem expression. One stage is intra-component optimization, namely distributedly optimizing {rate} and digital-analog resource allocation $R_i,(\alpha_i,K^d_i)$, based on given resource assignment $P_i,K_i$. And the other stage is inter-component optimization, which obtains the resource allocation vectors $\bm{P},\bm{K}$ among independent components, given the result of the first stage. Note that the process of the first stage could be viewed as the joint optimization for the single Gaussian source scenario (i.e., $m=1$). Hence, it could be solved asymptotically referring to the approach in Section \ref{section:single gaussian source}.
\par
In the subsequent discussion, we will focus on the analysis of the second stage. Since the problem of resource allocation among vectors is an MINLP problem, it is quite challenging to solve it directly. We tame the complexity by first relaxing the integer constraint. Thus, the slack problem can be formulated as follows
\begin{equation}\label{equation:muti-optimization2}
\begin{aligned}
\mathcal{P}4:~
&\min_{\bm{P},\bm{K}}
&& ED=\sum_{i=1}^mED_i\\
&~ \text{s.t.}
&& \sum_{i=1}^mK_i \leq K,~~ \sum_{i=1}^mP_i \leq P \\
&&& L < K_i,0 < P_i,~~ K_i, P_i\in\mathbb{R}.\\
\end{aligned}
\end{equation}
\begin{theorem}
{The slack optimization problem as formulated in (\ref{equation:muti-optimization2}) is a multivariate optimization problem, which is constrained convex for the vector $\bm{K}$ and $\bm{P}$.}
\end{theorem}
\begin{proof}
According to (\ref{equation:single-optimization2}), we can obtain the Hessian matrix of the function $ED_i$ with respect to $P_i$ and $K_i$. And the determinant of the Hessian matrix can be written as
\begin{equation}\label{equation:hessian matrix}
\begin{split}
H(ED_i)=&\frac{{\sigma_w^4}}{P_i^4 \alpha_i^2}2\ln^22\cdot 2^{R_{t}^t}({R_{i}^t})^2(2^{R_{i}^t}-1)\sigma^4_i\\
&-\frac{{\sigma_w^4}}{P_i^4 \alpha_i^2}{[(2^{R_{i}^t}-1)-2^{R_{i}^t}R_{i}^t\ln2]}^2\sigma^4_i\\
&+\frac{\pi e}{3}\ln^22 \frac{1}{P_i \alpha_i}\frac{1}{K_i-L}\frac{(1-\alpha_i)^2}{L^2{\  \sigma_w^2}}2^{R_{i}^t}({R_{i}^t})^22^{-2R_i}\\
&\cdot\sigma_i^4\int_{0}^{\infty}\frac{g_i^2}{(1+g_i\gamma^a_i)^3}p(g_i)\,dg_i.
\end{split}
\end{equation}
\par
Since $\gamma^a_i>0$ and $g_i \geq 0$, the integral part of the above function is non-negative. Then the residual part of $H(ED_i)$ can be represented as $\frac{N^2}{P_i^4 \alpha_i^2}\sigma^4_i\Upsilon(R^t_i)$, where the function $\Upsilon(x)$ can be expressed as:
\begin{equation}\label{equation:hessian matrix1}
\Upsilon(x)=2\ln^22\cdot 2^xx^2(2^x-1)-{[(2^x-1)-2^xx\ln2]}^2.
\end{equation}
\par
It can be proved that $\Upsilon(x)$ is a monotonically non-decreasing function when $x\geq0$, and $\Upsilon(0)=0$. Hence, based on the above analysis, $H(ED_i)$ is non-negative. Besides, it can be derived that $\frac{\partial^2 D_i}{\partial P_i^2}\geq0$ and $\frac{\partial^2 D_i}{\partial K_i^2}\geq0$. Thus, $D_i$ is a convex function for $P_i$ and $K_i$.
Further, we can verify that the problem as formulated in (\ref{equation:muti-optimization2}) is a constrained convex optimization problem for the vector $\bm{K}$ and $\bm{P}$ \cite{Boyd2014convex}.
\end{proof}
The above analysis suggests that the slack problem can be solved by general convex optimization methods, such as the interior-point method \cite{wright1997primal}, which is an efficient solution for a nonlinear constrained convex optimization problem. Based on the fractional results, we need to further obtain a feasible integer bandwidth solution. The rounding technique is first employed and then the residual channel uses are reallocated wisely.
\par
Intuitively, more channel uses should be allocated to the Gaussian vector with larger significance for better performance. If this is true, we may allow the Gaussian vector with larger variance to retain more channel uses at the reallocated step. Next, we will present a lemma to verify our conjecture.
\begin{lemma}
Give the power allocation vector $\bm{P}$, as well as rate allocation and digital-analog resource allocation strategies, to reduce system expected distortion, the bandwidth allocation vector must satisfy: the Gaussian vector with larger variance is allocated with more channel uses than the vector with smaller variance. In other words, suppose the set of variances denoted by ($\sigma_1^2,...,\sigma_m^2$) is in the descending order (i.e. $\sigma_1^2\geq\sigma_2^2\geq...\geq\sigma_m^2$), and the corresponding bandwidth allocation solution is $\bm{K}^*=(K_1^*,...K_m^*)$. Then, $\bm{K}^*$ should satisfy:
\begin{equation}\label{equation:bandwidth-optimal-condition}
K_1^*\geq K_2^*\geq...\geq K_m^*.
\end{equation}
\end{lemma}
\begin{proof}
According to the analysis in \ref{subsection:Problem Analysis}, the component of expected system distortion, which is associated with the results of bandwidth allocation, is $\sum_{i=1}^{i=M}\frac{2^{R_{i}^t}-1}{\gamma^d_i}\sigma^2_i$.
\par
Define $\Lambda(K_i^*)=\frac{2^{R_{i}^t}-1}{\gamma^d_i}=\frac{(K_i^*-L){\  \sigma_w^2}}{\alpha_iP_i}(2^{\frac{L}{K_i^*-L}R_i}-1)$. {\  In fact, this can be formulated into the rearrangement inequality problem \cite{hardy1994inequalities}. Note that $\sigma_1^2\geq\sigma_2^2\geq...\geq\sigma_m^2$. Thus, according to \cite{hardy1994inequalities}, to minimize the above distortion, the values of $\Lambda$ should be in the ascending order, i.e.,}
\begin{equation}\label{equation:ascending-order}
\Lambda(K_1^*)\leq\Lambda(K_2^*)\leq...\leq\Lambda(K_m^*).
\end{equation}
Define $y(x)=b(2^{\frac{a}{x}}-1)x$, where $a,b$ are positive constants and $x>0$. When $a=LR_i$ and $b=\frac{{\  \sigma_w^2}}{\alpha_iP_i}$, then $\Lambda(K_i^*)=y(K_i^*-L)$.
The {\  derivative} of $y(x)$ can be obtained as $y'(x)=b(2^{\frac{a}{x}}-1-\frac{a}{x}\ln2\cdot2^{\frac{a}{x}})$. Let $t=\frac{a}{x}>0$ and $g(t)=2^t-1-t\ln2\cdot2^t$, further we can derive that $g(t)<\lim_{t\to0}g(t)=0$. Hence, $y(x)$ is a monotonically decreasing function, which means (\ref{equation:ascending-order}) is equivalent to (\ref{equation:bandwidth-optimal-condition}).
\end{proof}
Lemma 1 gives an efficient greedy guideline for bandwidth reallocation after rounding the fractional results.
\subsection{Solution}\label{subsection:muti-solution}
Based on the above analysis, we propose an efficient two-stage algorithm to solve such an MINLP resource allocation problem. The proposed algorithm utilizes the convexity of the slack problem, where variables $(K_1,...K_m)$ can range continuously. After getting the fractional values, we further integrate rounding and greedy strategies to reallocate the bandwidth for a feasible integer solution. The algorithm is elaborated in Algorithm \ref{alg:Optimal Strategy for multivariate Gaussian Source}.
\par
In Algorithm \ref{alg:Optimal Strategy for multivariate Gaussian Source}, we first employ the interior-point method, which can guarantee the convergence of  the optimal solution to the convex slack problem (Lines 2-8). Define the penalty function of $\mathcal{P}4$ in the $j^{th}$ iteration as $\Phi(\bm{P},\bm{K},e^{(j)})$, which can be expressed as:
\begin{equation}\label{equation:penalty function}
\Phi(\bm{P},\bm{K},e^{(j)})=ED(\bm{P},\bm{K})-e^{(j)}\sum_{u=1}^{2m+2}\varphi (g_u(\bm{P},\bm{K})),
\end{equation}
where $g_u(\bm{P},\bm{K})\leq 0$ is the $u^{th}$ constraint function of $\mathcal{P}4$. $\varphi(x)$ is the penalty function, which has two alternatives, i.e., $\varphi(x) = 1/x$ or $\varphi(x) = \ln (-x)$. $e^{(j)}$ is the penalty factor of the $j^{th}$ iteration, which should be a positive number. The resource allocation vector is evolved from the initial value of equal power and bandwidth allocation. During the iteration, the value of the penalty factor is updated as $e^{(j)}=Ce^{(j-1)}$. The choice of $C$ should be between 0 and 1 to satisfy that $e^{(j)}$ is in the descending order with iterations and $\lim\limits_{j\to\infty} e^{(j)}=0$, which could guarantee the convergence of the algorithm. The iteration won't stop until the change of expected system distortion falls beneath a certain threshold.
\par

\begin{algorithm}
  \caption{Resource Allocation Strategy among Vectors}
   \label{alg:Optimal Strategy for multivariate Gaussian Source}
   \KwIn{$\Sigma_S,~\{R_i\},~\{\alpha_i\},~\{K_i^d\},~K,~P,~\sigma_w^2,~\epsilon,~\tilde e,~C$}
   \KwOut{$\bm{P}^*=(P^*_1,...,P^*_m),~\bm{K}^*=(K^*_1,...,K^*_m)$}
   \textbf{initialize procedure:}\\
   $e^{(0)}=\tilde e,~j=0$\;
   $P^{(0)}_i=\frac{1}{m}P,~K^{(0)}_i=\frac{1}{m}K,~(i=1,2,...,m)$\;
   \SetKw{Kwend}{end initial procedure}

   \Repeat{$ED(\bm{P}^{(j)},~\bm{K}^{(j)})-ED(\bm{P}^{(j-1)},~\bm{K}^{(j-1)}) \le \epsilon$}{
   $j = j+1,~~~e^{(j)}=Ce^{(j-1)}$\;
   $\Phi(\bm{P}\!,\bm{K},e^{(j)})\!=\!ED(\bm{P},\bm{K})\!-\!e^{(j)}\!\!\sum\limits_{u=1}^{2m+2}\!\!\varphi (g_u(\bm{P},\bm{K}))$\;
   $(\bm{P}^{(j)},~\bm{K}^{(j)})=\arg\min~\Phi(\bm{P},\bm{K},e^{(j)})$\;}

   $\bm{P}^*=\bm{P}^{(j)}$\;
   Round $\bm{K}^*=\lfloor\bm{K}^{(j)}\rfloor$ to obtain a basic solution\;
   $K_r=K-\rm{sum}(\bm{K^*})$\;

   \Repeat{$K_r=0$}{
   $K_r = K_r-1$\;
   $\mathcal{S}_r=\{1\}\cup\{n\mid K_n^*<K_{n-1}^*\}$\;
   \ForEach {$~n_r \in \mathcal{S}_r$ }{
   $\bm{K} = \bm{K^*}, ~\rm{except} ~{K_{n_r}}= {K_{n_r}^*}+1$\;
   Calculate $ED$ according to (\ref{equation:single-optimization2}),({\ref{equation:muti-optimization1}})\;
   }
   $ED_r=\min\{ED\}$\;
   $\bm{K^*}=\bm{K} ~{\rm corresponding~to} ~ED_r$\;}

   \Return $\bm{P}^*,~\bm{K}^*$\;
\end{algorithm}
\par
Based on the fractional bandwidth results and inspired by Lemma 1, we integrate rounding and greedy strategies to reallocate the bandwidth for a feasible integer solution (Lines 10-21). The continuous values are rounded down for a basic bandwidth allocation solution in Line 10. Then, the residual channel uses are further reallocated. Based on Lemma 1, we only need to consider a few candidates in the set of $\mathcal{S}_r$ during each reallocation. This decreases the number of possible locations to reallocate the channel use and assures that the bandwidth solution after reallocation satisfies (\ref{equation:bandwidth-optimal-condition}). Inspired by Lemma 1, the locally optimal solution can be found in this greedy manner.
\par
{\  Now, we can analyze the computational complexity of Algorithm 2. At the first stage, it requires complexity of $\mathcal{O}(\frac{1}{C}\ln\frac{4N\tilde{e}}{\epsilon})$ \cite{hertog1994interior}. Here, $N=2m+2$ is the number of constraint functions. At the second stage, it requires complexity of $\mathcal{O}(mK_r)$. Hence, the complexity of Algorithm 2 is $\mathcal{O}(\frac{1}{C}\ln\frac{4N\tilde{e}}{\epsilon}+mK_r)$.}
\par
Combined with the solution of the first stage adopting Algorithm \ref{alg:Optimal Strategy for Single-Gaussian Source}, we can finally approach the solution of the optimization problem formulated in (\ref{equation:muti-optimization1}) for the multivariate Gaussian source scenario. In fact, the idea of the proposed two-stage algorithm can also be traced in the BCD method.
\par
\section{Performance Evaluation}\label{section:simulation}
We carry out simulations to evaluate performance of HDA transmission where joint rate and resource allocation is performed under various scenarios. The channel fading coefficient of the $i^{th}$ vector ($\bm S_i$) transmission is generated as $h_i\sim\mathcal{CN}(0,1)$. All simulations are implemented in Matlab 2015a. For the multivariate Gaussian source scenario, transmission of any vector component is considered to be independent with each other. Signal-to-distortion ratio (SDR) \cite{skoglund2006hybrid} in dB is used as the evaluation metric, which is defined as:
\begin{equation*}\label{equation:SDR}
    SDR\buildrel\Delta\over=10\lg\frac{\sigma_i^2}{D_i}=10\lg\frac{{E}||\bm{S}_i||^2}{{E}||\bm{S}_i-\bm{\hat{S}}_i||^2}.
\end{equation*}
\par

\subsection{Simulation Results of Single Gaussian Source Scenario}\label{subsection:single-simulation}
In this subsection, we evaluate performance of HDA transmission for single Gaussian vector $\bm S_i$. The i.i.d. complex Gaussian vector $\bm S_i$ is with zero mean and unit variance.
\subsubsection{Impacts of $\alpha_i$ and $R_i$}
We will evaluate the impacts of {rate} allocation and resource allocation on HDA performance. As analyzed in \ref{subsection:Problem Analysis}, bandwidth allocation is determined by the dimension of the vector. Thus, we only focus on {rate} allocation represented by the quantization rate $R_i$, and power allocation represented by the coefficient $\alpha_i$. Scenarios with various average SNR settings defined in (\ref{equation:generalSNR}), as well as various bandwidth expansion ratio settings defined as $\eta=\frac{K_i}{L}$ channel uses/sample, are took into consideration.
\par
\begin{figure}[ht]
\centering
\subfigure[$\eta =2$]{\includegraphics[width=4.35cm]{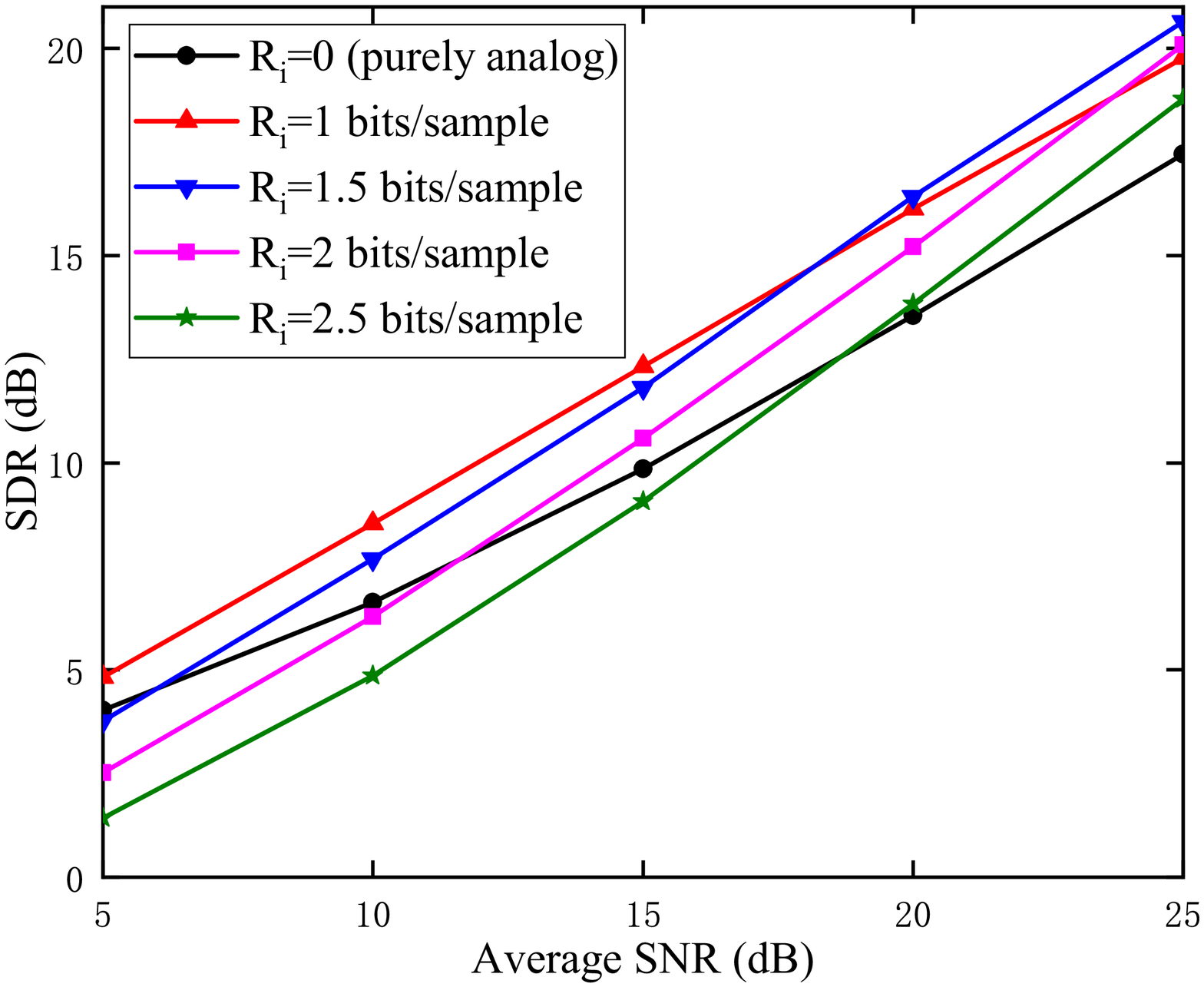}\label{figure:single_fix_power_SDR_a}}
\subfigure[$\eta =3$]{\includegraphics[width=4.35cm]{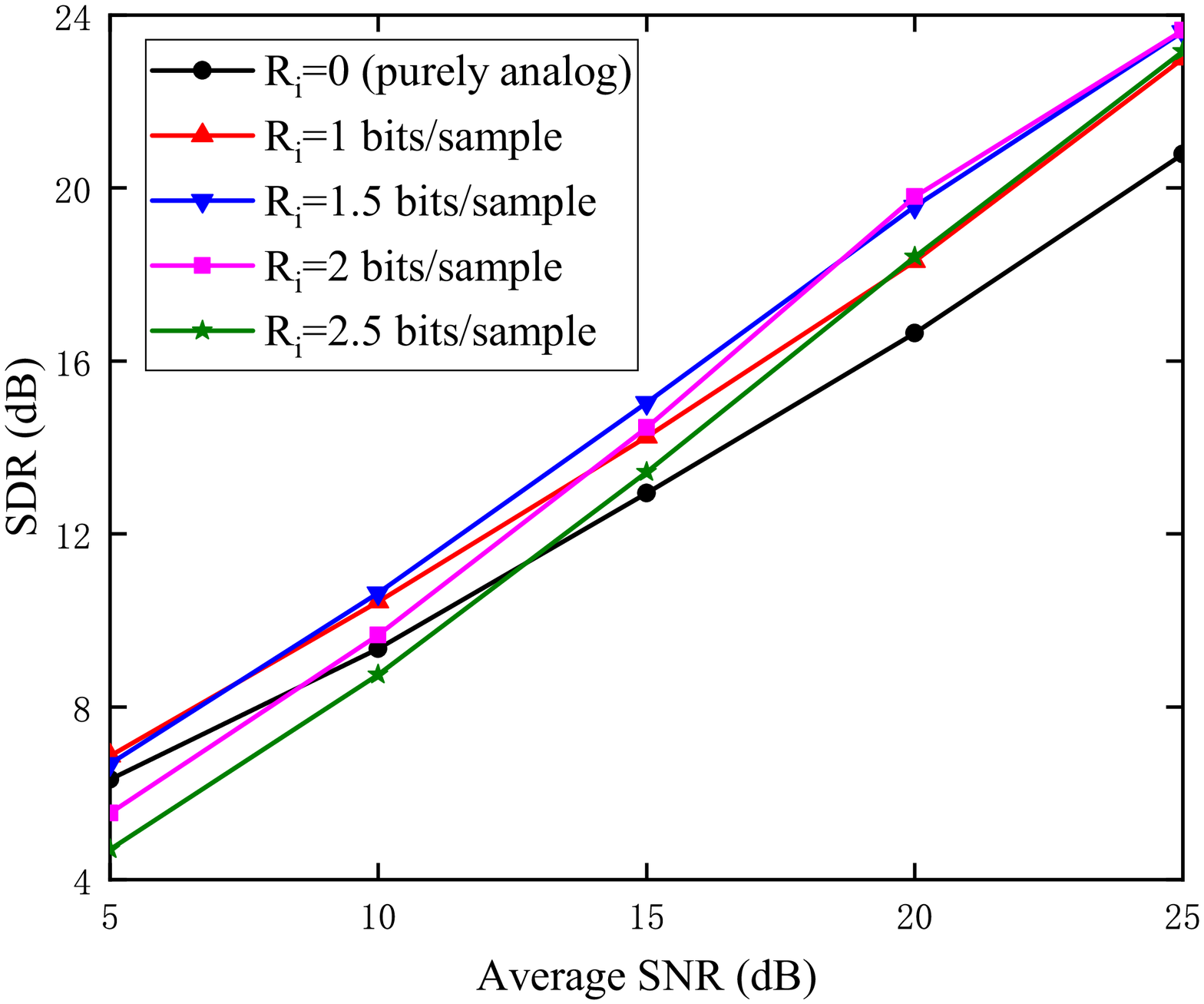}\label{figure:single_fix_power_SDR_b}}
\caption{SDR performance of the HDA system for various quantization rates $R_i$; $\alpha_i=0.6$; single Gaussian vector with unit variance over the quasi-static Rayleigh fading channel.}\label{figure:single_fix_power_SDR}
\end{figure}
We first evaluate HDA performance with various quantization rates in the case of $\alpha_i=0.6$, as shown in Fig.\ref{figure:single_fix_power_SDR}. Over a wide range of SNR, performance with various $R_i$ alternating rises. Specifically, at low SNR, the HDA system with smaller $R_i$ tends to provide better performance. Purely analog transmission, corresponding to $R_i=0$ also performs better at low SNR. At medium and high SNR, the HDA system with larger $R_i$ would provide performance improvement. Besides, the system with much large $R_i$ (e.g., $R_i=2.5$ bits/sample) would show substantially improved performance with higher $\eta$ (e.g., $\eta=3$ channel uses/sample). The reasons for these observation will be elaborated later.
\par
Then, we fix $R_i=1.5$ bits/sample and implement the HDA system with various power allocation coefficients $\alpha_i$, as shown in Fig. \ref{figure:single_fix_rate_SDR}. The extreme cases of $\alpha_i=0$ and $\alpha_i=1$ correspond to purely analog transmission and purely digital transmission, respectively. From Fig. \ref{figure:single_fix_rate_SDR}, the SDR curve of $\alpha_i=0.6$ is above other SDR curves for almost SNR settings. This implies that the optimal $\alpha_i$ is about 0.6 when $R_i=1.5$ bits/sample. Another observation is that the system with small $\alpha_i$ (e.g., $\alpha_i=0.3$) is inferior to that with larger $\alpha_i$ (e.g., $\alpha_i=0.9$) at low and medium SNR. However, reverse results can be observed at high SNR. Moreover, the turning point lies at lower SNR when $\eta=3$ channel uses/sample.
\par
\begin{figure}[ht]
\centering
\subfigure[$\eta =2$]{\includegraphics[width=4.35cm]{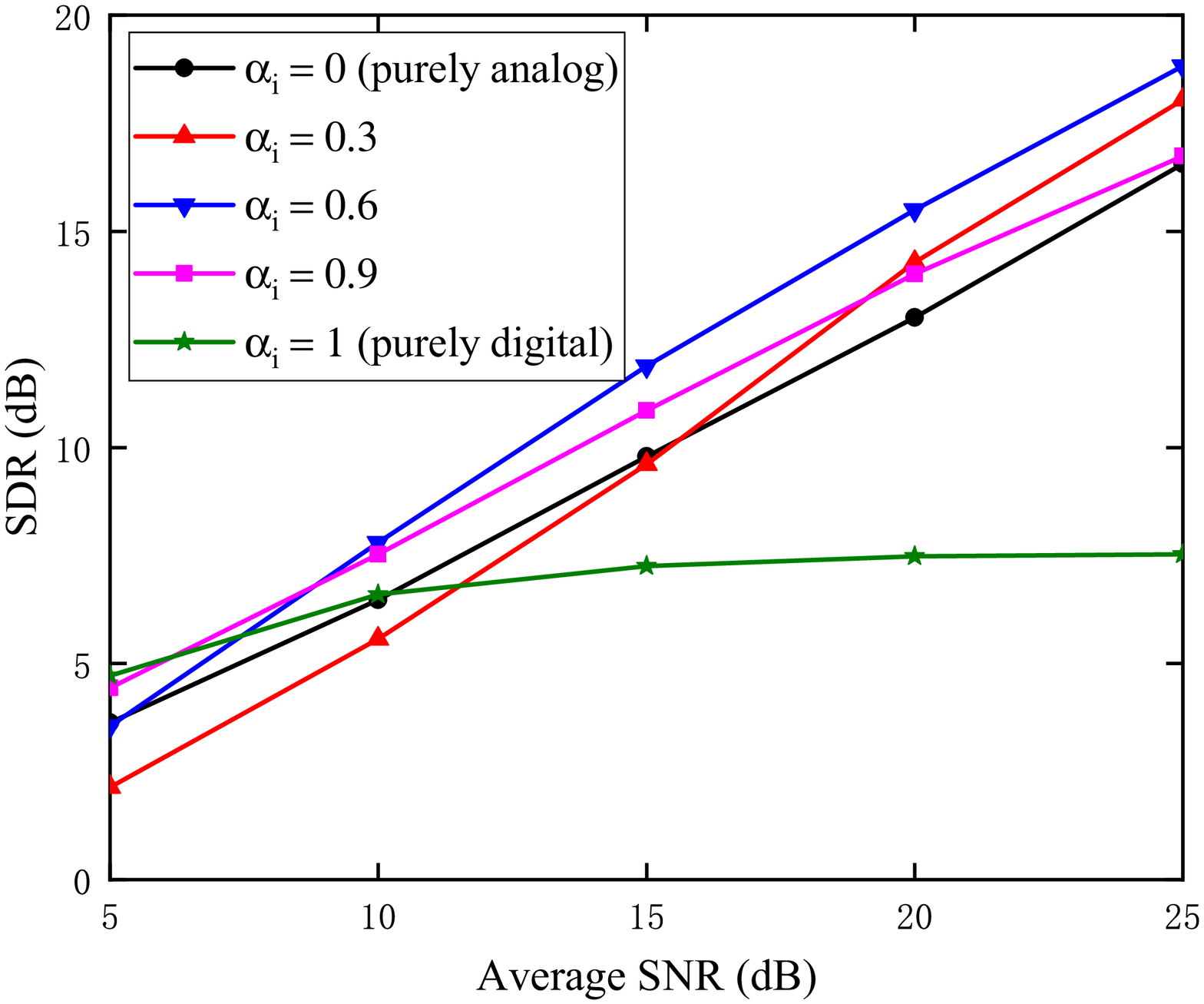}\label{figure:single_fix_rate_SDR_a}}
\subfigure[$\eta =3$]{\includegraphics[width=4.35cm]{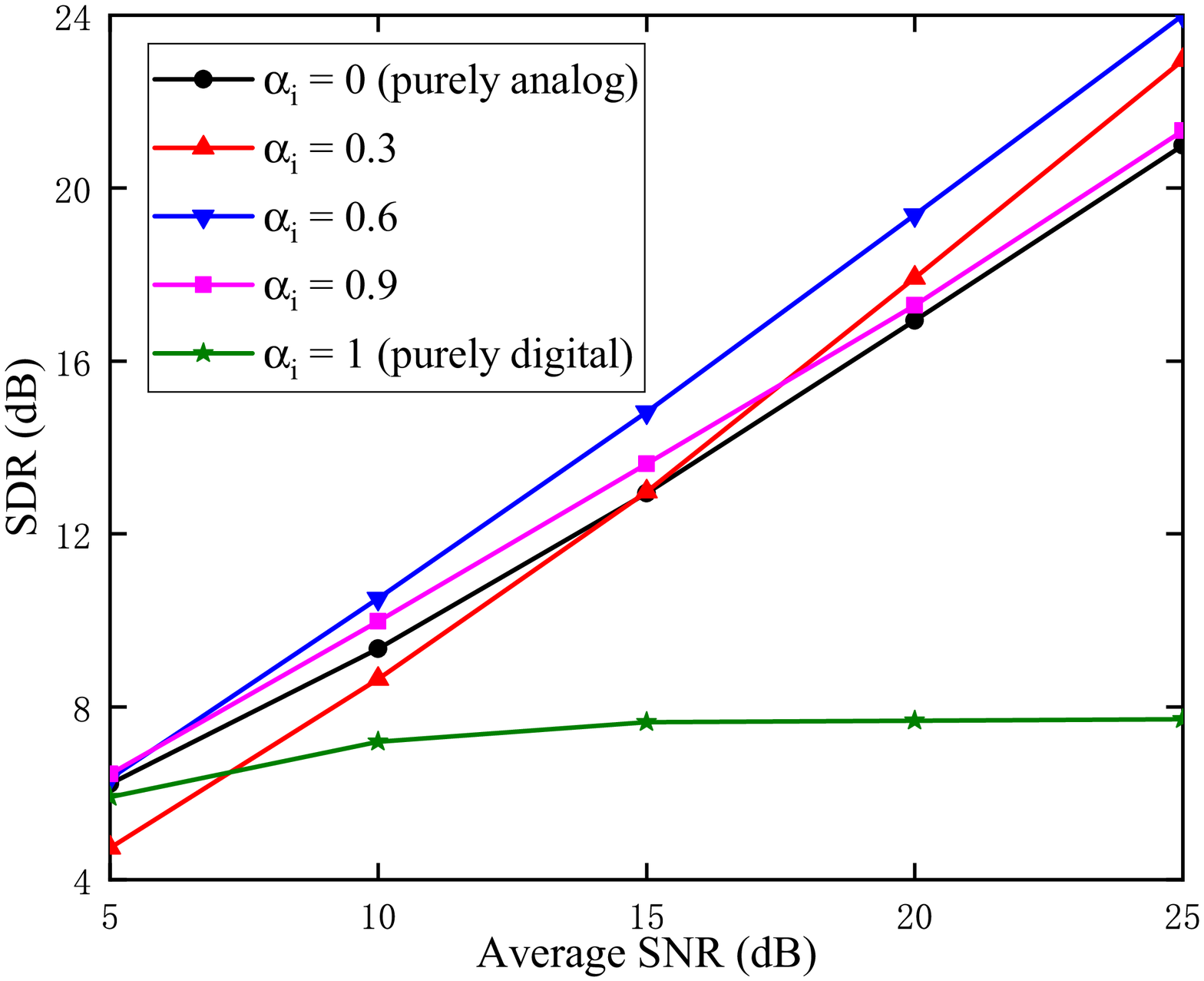}\label{figure:single_fix_rate_SDR_b}}
\caption{SDR performance of the HDA system for various power allocation coefficients $\alpha_i$; $R_i=1.5$ bits/sample; single Gaussian vector with unit variance over the quasi-static Rayleigh fading channel.}\label{figure:single_fix_rate_SDR}
\vspace{-0.2cm}
\end{figure}

\par
\begin{figure}[ht]
\centering
\subfigure[The results of quantization rate]{\includegraphics[width=4.35cm]{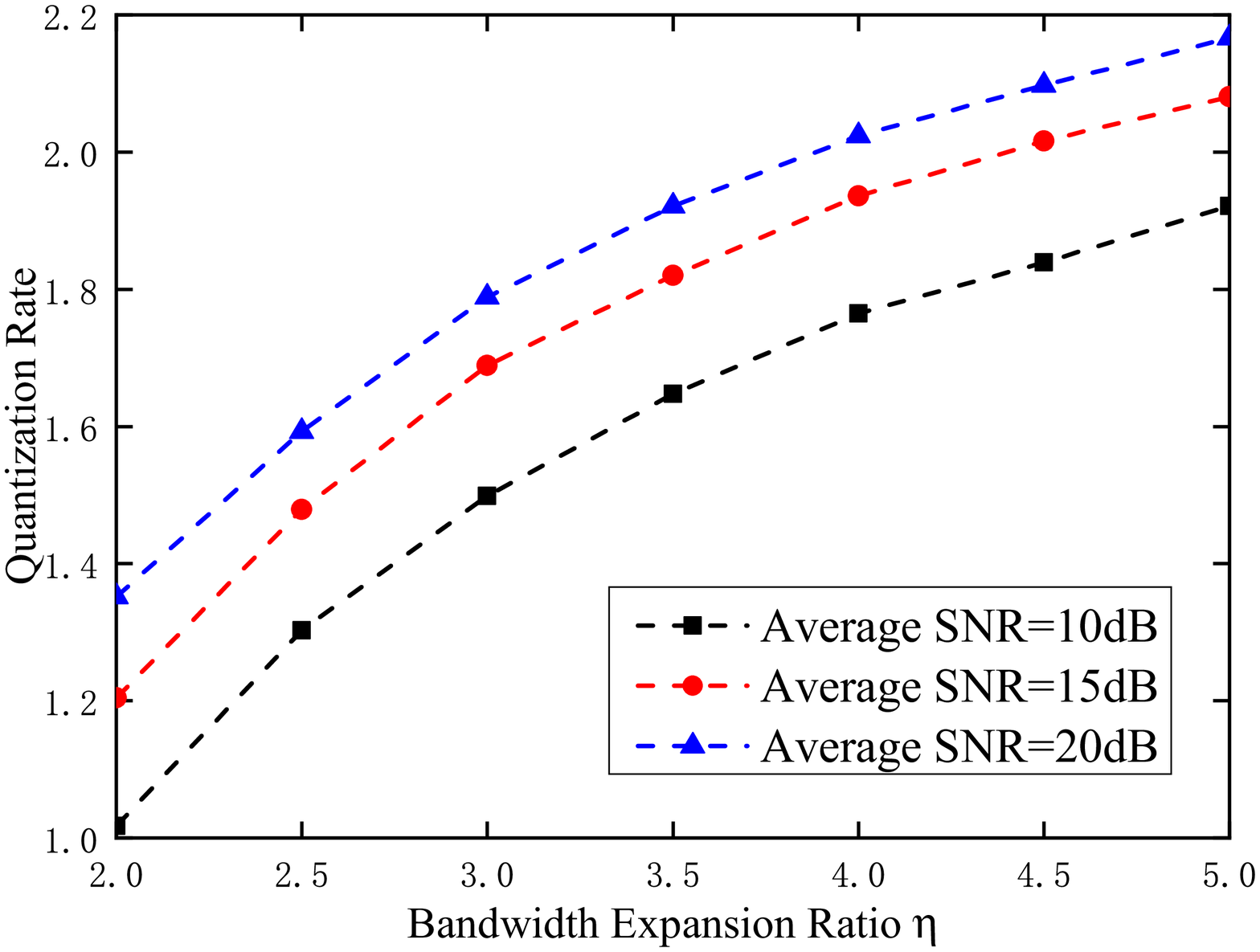}}
\subfigure[The results of power allocation]{\includegraphics[width=4.35cm]{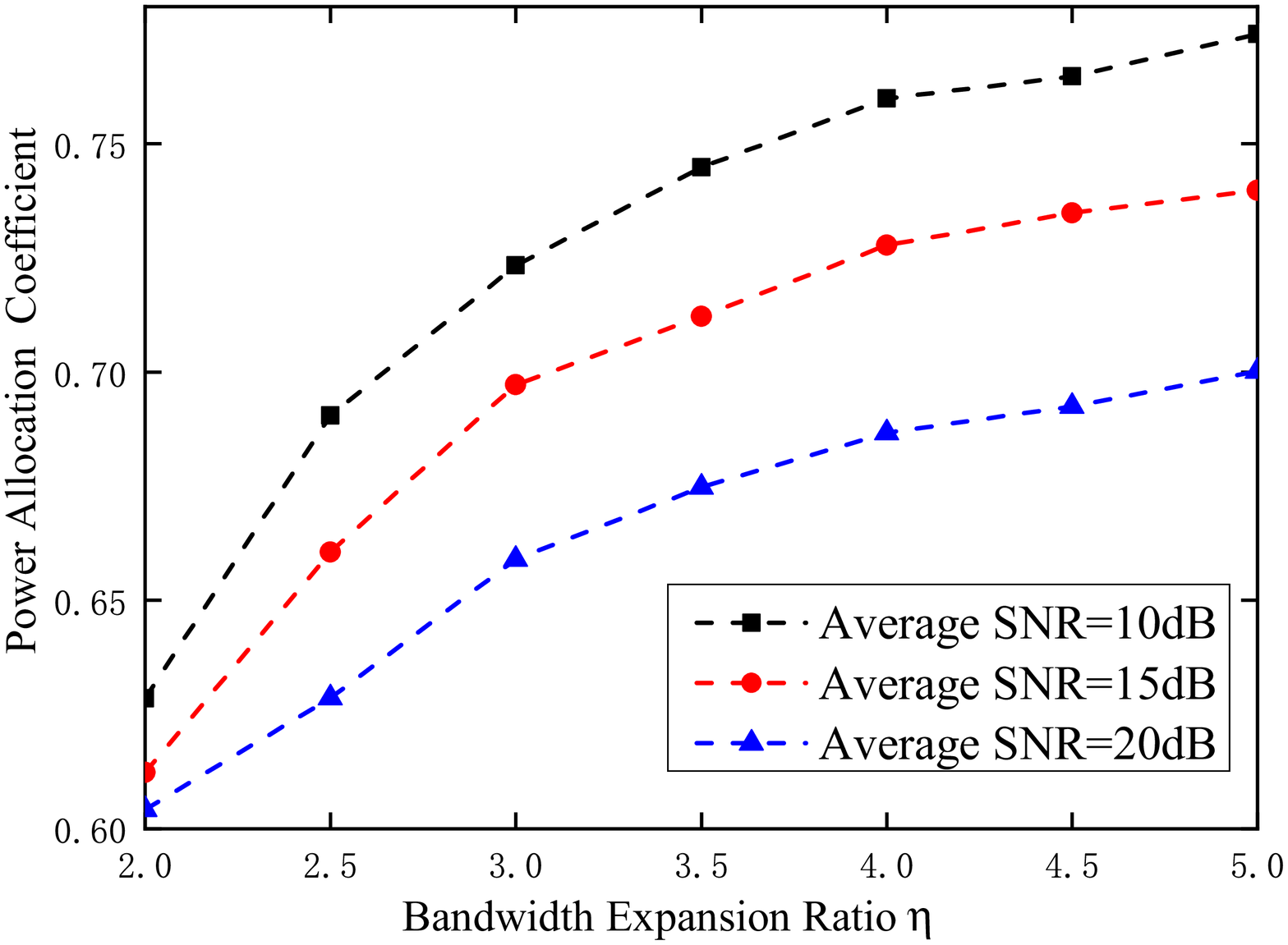}}
\caption{The asymptotical optimization results $R_i$ and $\alpha_i$ of the HDA system; single Gaussian vector with unit variance over the quasi-static Rayleigh fading channel.}\label{figure:optimal-parameters}
\end{figure}
The above observations indicate that $R_i$ and $\alpha_i$ should be selected wisely considering SNR and bandwidth expansion ratio $\eta$. To further explain these observations, we depict the joint allocation results versus $\eta$ at SNR=10\,dB, 15\,dB and 20\,dB. The results in Fig.\ref{figure:optimal-parameters} illustrate several insights, which are consistent with above observations. First, higher $R_i$ and larger $\alpha_i$ are required for better performance with increasing $\eta$. Such phenomenon can be explained as follows. When more channel uses are available, higher transmitting rate can be achieved with acceptable outage probability. Meanwhile, this implies more {information} is conveyed by digital transmission, which should be assigned more power. Second, higher SNR would result in higher R and smaller $\alpha_i$. In fact, increasing SNR indicates better channel conditions with greater channel capacity. Thus, the HDA system can be implemented at higher rate. In this case, smaller digital power could also ensure lower outage probability, and more residual power can be allocated for analog transmission to improve system performance.

\subsubsection{Performance comparison}
We compare performance of the proposed HDA system against other systems, as shown in Fig.\ref{figure:single_compare-SDR}. The realistic scenario are considered, where the actual noise power is not known to the encoder. Thus, the encoders of all schemes are optimized at a target channel SNR $\gamma^{tar}=10$\,dB. The reference systems are described as follows.
\begin{itemize}
\item Purely Analog Scheme: a purely analog scheme solely employing the analog part of the HDA scheme ($\alpha_i=0$). Note that analog transmission only occupies $L$ channel uses, while another $(\eta-1)L$ channel uses are silent. For fair comparison, the transmitter is allowed to scale its power to $\eta P$, as that in \cite{caire2007distortion}, \cite{gunduz2005source}.
\item Purely Digital Scheme: a purely digital scheme solely employing the digital part of the HDA scheme ($\alpha_i=1$).
\item Chen-Tuncel Scheme: a HDA scheme designed for the Wyner-Ziv problem proposed in \cite{chen2014zero}. As analyzed in \cite{chen2014zero}, the designed system can be adopted in HDA transmission for 1:2 bandwidth expansion without side information. Specifically, the scheme transmits the uncoded vector in the first channel use, as well as the superposed digital and analog signals in the next channel use. Both power allocation of the two channel uses and digital-analog power allocation of the second channel use are optimized;
\item Theoretical Limit (OPTA): the optimal performance theoretically attainable (OPTA) for quasi-static Rayleigh fading channels, given by $D_i^{opt}=\frac{e^{1/\gamma_i}}{\gamma_i}\int_{1}^{\infty}\frac{e^{-t/\gamma_i}}{t^\eta}\,dt$ \cite{gunduz2005source}.
\end{itemize}
\par
As Fig.\ref{figure:single_compare-SDR} shows, the proposed HDA system outperforms the purely digital system and the purely analog system for almost SNR settings. Although the purely digital system suffers from cliff effect, the proposed HDA system can provide robust performance over a whole range of SNR, and alleviate cliff effect apparently (always with positive slope of SDR curve). Moreover, the proposed HDA system can achieve performance gains between 0$\sim$2.3\,dB in terms of SDR, compared with the Chen-Tuncel scheme. This is due to the joint rate allocation and resource allocation, with consideration of fading distribution.
\par
\begin{figure}[hbt]
\centering
\includegraphics[width=6cm]{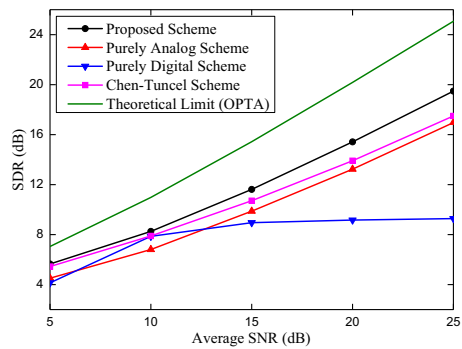}
\caption{SDR performance of the optimized HDA system, the Chen-Tuncel system, the purely analog, the purely digital systems and OPTA; single Gaussian vector with unit variance over the quasi-static Rayleigh fading channel, $\eta = 2$ channel uses/sample, $\gamma^{tar}=10$\,dB.}\label{figure:single_compare-SDR}
\end{figure}

Note that in the problem formulation in Section \ref{subsection:Problem Analysis}, approximation is adopted with small $\tau_i\rightarrow0$, which implies the limit of low outage probability. To assess the effect of such approximation, we evaluate the HDA scheme optimized with mesh grid method. Here, $\alpha_i$ and $R_i$ are searched to minimizing system distortion as formulated in (\ref{equation:distortionExpression1}) without approximation. Besides, in the problem analysis in Section \ref{subsection:Problem Analysis}, the asymptotical characteristic of ECSQ as formulated in (\ref{equation:rate-distortion}) is utilized. For the integrity of the paper, we evaluate the performance of the proposed scheme with practical ECSQ and entropy coding implementations. We employ the fixed-point algorithm \cite{netravali1976optimum} to design the ECSQ (FP-ECSQ) for quantizing a complex Gaussian vector with $L=3000$ samples. As suggested in \cite{netravali1976optimum}, the symmetric quantizer is considered. Then Huffman coding is adopted as the entropy coding to encoding the quantized bits after the ECSQ. Considering implementation complexity, the scheme is designed without channel codes in the digital part, as that in \cite{wang2009hybrid}. The results are depicted in Fig. \ref{figure:single_mesh}. The average bandwidth expansion ration $\eta$ is set as 3 channel uses/sample. The encoders of all schemes are designed with actual channel SNR. As Fig. \ref{figure:single_mesh} shows, the gap between the performance with mesh grid method and the asymptotic performance in 5 dB is larger than that in higher SNR. This is due to that when SNR is 5 dB, the channel condition cannot support reliable transmission with low outage probability. However, for most channel SNR settings, small gaps exist among these three curves of simulation results, which suggests that the approximation adopted in the analysis is sensible, and the solution is good to be utilized for practical ECSQ.
\par

\begin{figure}[hbt]
\centering
\includegraphics[width=6cm]{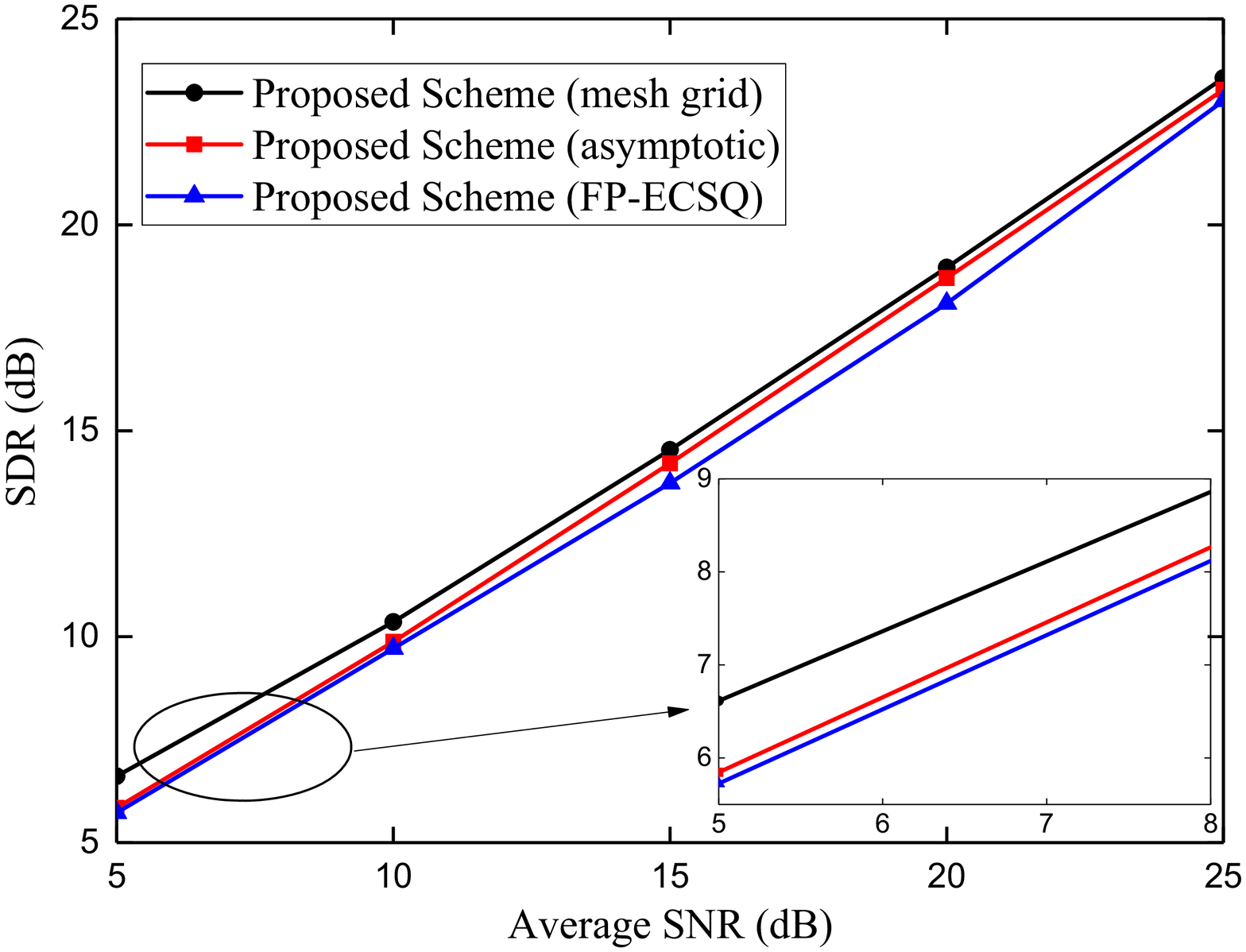}
\caption{SDR performance of the scheme optimized with mesh grid, the scheme optimized asymptotically, and the scheme implemented with FP-ECSQ; single Gaussian vector with unit variance over the quasi-static Rayleigh fading channel, $\eta = 3$ channel uses/sample.}\label{figure:single_mesh}
\vspace{-0.5cm}
\end{figure}
\par
\subsection{Simulation Results of Multivariate Gaussian Source Scenario}\label{subsection:multi-simulation}
Next we evaluate performance of HDA transmission for the multivariate Gaussian vectors with an $m \times m$ diagonal covariance matrix $\Sigma_S={\rm diag} \, (\sigma_1^2,\sigma_2^2,...,\sigma_m^2)$. In the following simulations, the number of vectors is set to 4, i.e., $m=4$. The corresponding diagonal entries are set as $(2,1,0.5,0.1)$ similar to the differences among sensor measurements in \cite{liu2012energy}. In fact, the difference of variances among different vectors may be of larger magnitude in the case of video transmission \cite{cui2014robust}.
\par

\subsubsection{{\ Impacts of intra-component optimization and inter-component optimization}}
Here we implement the experiment to validate the conclusions in Section \ref{section:multi-gaussian source}. As analyzed previously, the optimization of multivariate Gaussian vectors is addressed in two stages, i.e. intra-component optimization and inter-component optimization. Fig.\ref{figure:multi-Gaussian_singleopt SDR} presents the separate effect of the optimization over each stage. For intra-component optimization, ``ERA'' means the resource is allocated equally among vectors. For inter-component optimization, two settings for the rate and the power allocation coefficient are considered, i.e., ($R_i=1.5, \alpha_i=0.4)$ and ($R_i=2, \alpha_i=0.6)$ for $i=1,\cdots,4$. We can observe that both intra-component optimization and inter-component optimization are of great effects on overall performance. We also remark that the proposed system always outperform the separately optimized systems and the gain is substantial for entire SNR settings (about 1$\sim$2.5dB). Thus the joint optimization is essential to provide graceful and robust performance.
\begin{figure}[ht]
\centering
\includegraphics[width=6cm]{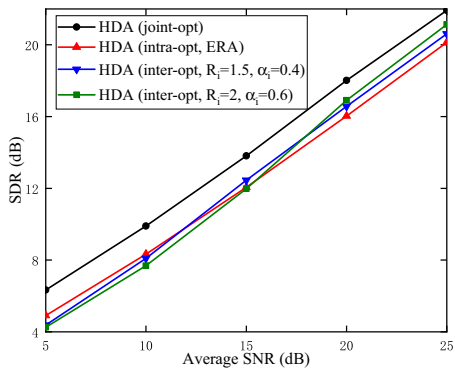}
\caption{SDR performance of HDA system with joint optimization, HDA scheme with intra-component optimization, and HDA scheme with inter-component optimization; multivariate Gaussian vectors over the quasi-static Rayleigh fading channel, $\eta = 2$ channel uses/sample, $\gamma_{tar}=10$\,dB.}\label{figure:multi-Gaussian_singleopt SDR}
\end{figure}

\par

\begin{figure}[ht]
\centering
\includegraphics[width=6cm]{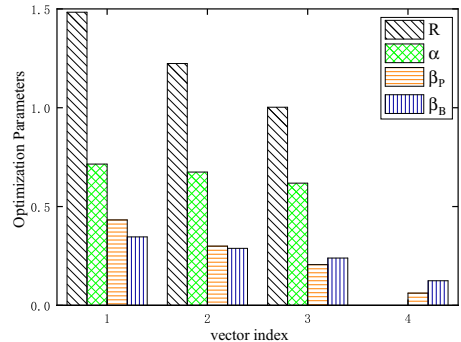}
\caption{The joint optimization results of the proposed HDA system; multivariate Gaussian vectors over the quasi-static Rayleigh fading channel, $\eta = 2$ channel uses/sample, $\gamma_{tar}=10$\,dB.}\label{figure:multi_parameter}
\end{figure}


\par
We next depict the joint optimization results for HDA transmission of the multivariate Gaussian vectors in Fig.\ref{figure:multi_parameter}. Due to space limit, we only present the result optimized at a target SNR setting $\gamma^{tar}=10$\,dB with a bandwidth expansion ratio $\eta=2$ channel uses/sample. Similar results can be obtained for other SNR values and other bandwidth expansion scenarios. The vectors are indexed from 1 to 4 in the descending order of variances. $\bm{\beta_B}$ and $\bm{\beta_P}$ are the ratios of bandwidth and power allocation among vectors, respectively, which can be expressed as ${\beta_B}_i=\frac{K_i}{K},{\beta_P}_i=\frac{P_i}{P}$. From Fig.\ref{figure:multi_parameter}, we could figure out that the vector with greater variance should be assigned more transmission resources and quantized with a higher rate. Meanwhile, more power should be allocated for digital transmission. Such observation supports our previous claims in Section \ref{section:multi-gaussian source} that the priority (i.e., variance) of the Gaussian vector plays an important role in the joint {rate} allocation and resource allocation. Moreover, it also validate our conjecture in Lemma 1. Particularly, the digital-analog power allocation coefficient and the quantization rate of the $4^{th}$ vector are zero, which means its transmission reduced to the purely analog type. The reason of this result is the fact that the vector with lower priority saves its bandwidth resource for transmission of the vector with higher priority.
\begin{figure}[ht]
\centering
\includegraphics[width=6cm]{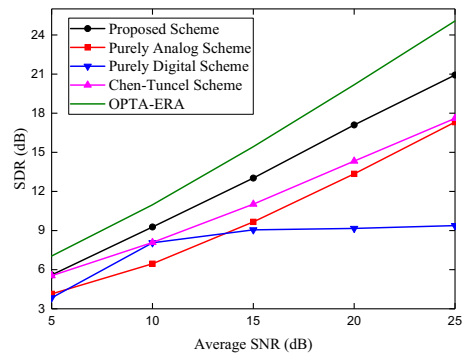}
\caption{SDR performance of the optimized HDA system, the Chen-Tuncel system, the purely digital and the purely analog systems; multivariate Gaussian vectors over the quasi-static Rayleigh fading channel, $\eta = 2$ channel uses/sample, $\gamma_{tar}=10$\,dB.}\label{figure:multi_compare-SDR}
\end{figure}

\begin{figure}[hbt]
\centering
\includegraphics[width=6cm]{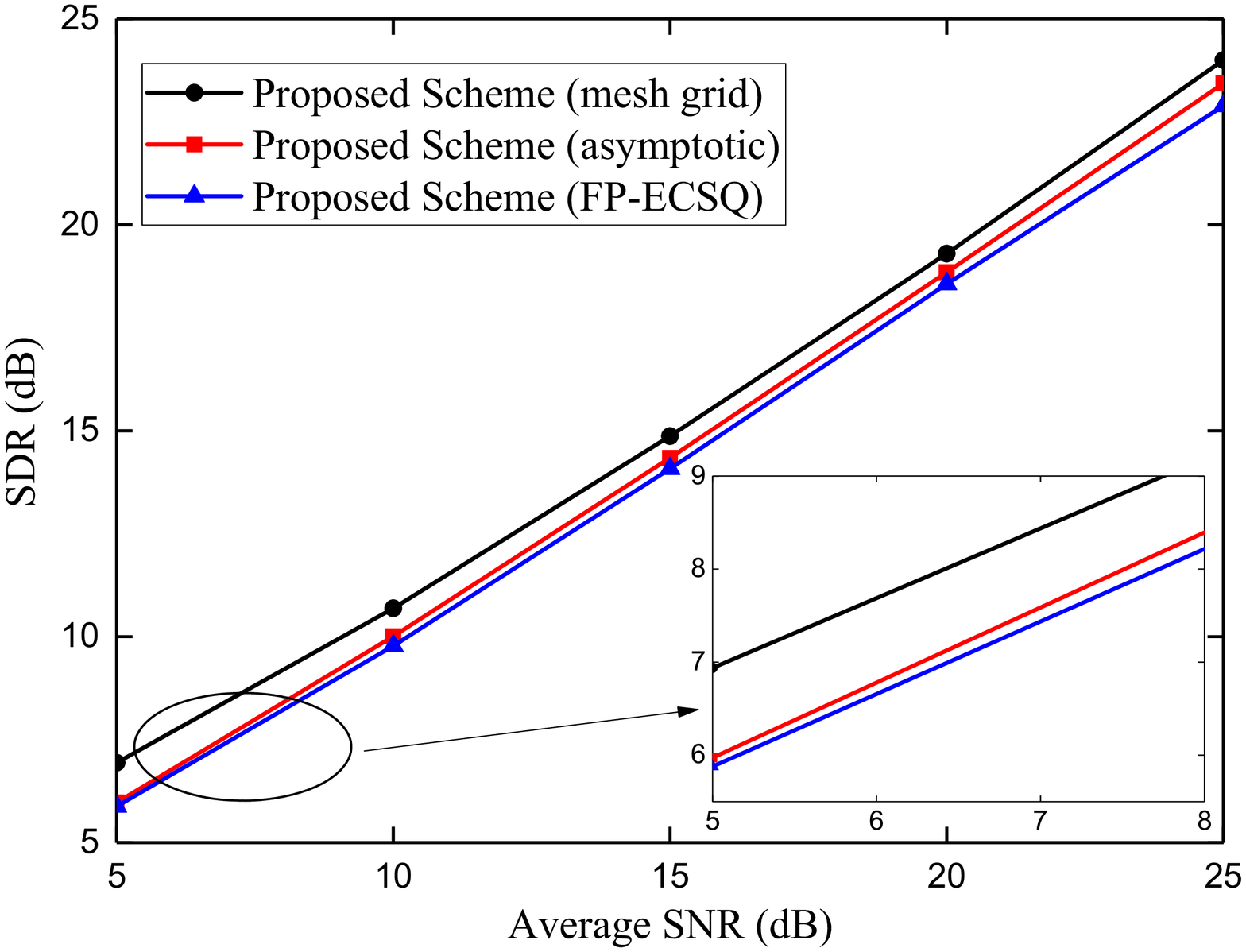}
\caption{SDR performance of the scheme optimized with mesh grid, the scheme optimized asymptotically, and the scheme implemented with FP-ECSQ; multivariate Gaussian vectors over the quasi-static Rayleigh fading channel, $\eta = 3$ channel uses/sample.}\label{figure:multi_compare-SDR}
\end{figure}
\subsubsection{Performance comparison}
Then we compare the proposed HDA system with other systems over various SNR settings at $\eta=2$ channel uses/sample, as illustrated in Fig.\ref{figure:multi_compare-SDR}. The comparisons are the purely analog, the purely digital and the Chen-Tuncel systems. All the reference schemes are implemented with the ERA strategy among vectors. And the theoretical limit is derived for the ERA scheme. Besides, the encoders of all schemes are optimized at $\gamma^{tar}=10$\,dB. From Fig.\ref{figure:multi_compare-SDR}, we remark that systems in the multivariate Gaussian source scenario perform analogously to those in single Gaussian source scenario (cf. Fig.\ref{figure:single_compare-SDR}). However, the SDR gain of the proposed system is enhanced considerably (up to 3.5dB) due to the joint optimization with consideration of source characteristics differences. In fact, the SDR gain would be further improved when the difference of variances is of larger magnitude, according to our simulation results.
\par
For the integrity of the paper, we also evaluate performance of the HDA scheme optimized without approximation by the mesh grid method, and performance of the HDA scheme implemented with FP-ECSQ and Huffman coding, as Fig. \ref{figure:multi_compare-SDR} shows. Considering the implementation complexity, two vectors are selected as multivariate Gaussian vectors with the variance $\sigma_1^2=2, \sigma_2^2=1$ respectively. Compared with the result in Fig. \ref{figure:single_compare-SDR}, the gaps for the multivariate Gaussian source scenario are slightly larger. This is due to the fact that two vectors are encoded and the optimization of resource allocation between these two vectors are also need to be considered. Nevertheless, the result indicates that the proposed solution adopts sensible approximation in the analysis, and can be implemented for practical quantization schemes.

\section{Conclusion}\label{section:conclusion}
This paper has proposed a joint {rate} allocation and resource allocation scheme for HDA transmission with bandwidth expansion. Due to the considered quasi-static Rayleigh fading channel, outage occurs randomly. In this case, digital distortion cannot be neglected. With consideration of analog distortion as well as digital distortion, we have derived the expression of expected system distortion under given fading distribution. Based on our analytical work, the joint optimization problem for the single Gaussian source scenario has been formulated as a minimization problem on expected system distortion, which is hard to solve. We transform the original problem into an approximated problem, based on the limit of low outage probability and involved ECSQ. Then, inspired by the BCD method, a tractable iterative solution adopting the gradient method has been proposed to solve this problem asymptotically. Furthermore, we extend our work to the multivariate Gaussian source scenario. In this case, we propose a two-stage algorithm integrating rounding and greedy strategies. For the slack problem where bandwidth variables can range continuously, its convexity has been proven, which guarantees its optimal solution. Then the algorithm approaches the locally optimal integer solution in a greedy manner. Finally, extensive simulations are performed. The results validate the advantages of joint rate and resource allocation optimization under the single Gaussian source scenario as well as the multivariate Gaussian source scenario.
\par
The current discussion is based on the assumption of the Gaussian source. In the future, we plan to investigate the implementation of the proposed HDA transmission system into correlated Gauss-Markov source as well as practical multimedia applications.

%
%
%
\footnotesize
\balance

\end{document}